\documentclass[aoas,preprint]{imsart}

\RequirePackage{amsthm,amsmath,amsfonts,amssymb}
\RequirePackage[authoryear]{natbib}
\RequirePackage[colorlinks,citecolor=blue,urlcolor=blue]{hyperref}
\RequirePackage{graphicx}
\usepackage{booktabs}
\usepackage{array}
\usepackage[utf8]{inputenc}
\usepackage[english]{babel}
 \usepackage{amsthm}
 \usepackage[flushleft]{threeparttable}

 \usepackage{dsfont}
\usepackage{graphics,graphicx}
\usepackage{subfig}
\usepackage{tikz}
\usetikzlibrary{arrows,shapes, positioning}
\usepackage{ifthen}
\usepackage{url}
\usepackage{fancyhdr}
\usepackage{amssymb}
\usepackage{graphicx}
\usepackage{color}
\usepackage{colortbl}
\usepackage{adjustbox}
\usepackage[normalem]{ulem}
\usepackage{refcount}
\usepackage{float} 
\usepackage{makecell}

\newtheorem{assumption}{Assumption}
\newtheorem{definition}{Definition}
\newtheorem{corollary}{Corollary}

\startlocaldefs
\endlocaldefs

\begin{document}

\begin{frontmatter}
%\title{Bayesian causal forest for three ordered treatment levels: an application to nutritional epidemiology}
\title{
Bayesian Nonparametric Causal Inference for High-Dimensional Nutritional Data via Factor-Based Exposure Mapping
%option 2:\\
%Causal Effects of High-Dimensional Dietary Exposures: Factor-Based Mapping and Bayesian Nonparametric Modeling
}
\runtitle{Bayesian causal inference with exposure mapping}

\begin{aug}
%%%%%%%%%%%%%%%%%%%%%%%%%%%%%%%%%%%%%%%%%%%%%%
%%Only one address is permitted per author. %%
%%Only division, organization and e-mail is %%
%%included in the address.                  %%
%%Additional information can be included in %%
%%the Acknowledgments section if necessary. %%
%%%%%%%%%%%%%%%%%%%%%%%%%%%%%%%%%%%%%%%%%%%%%%

\author{\fnms{Dafne} \snm{Zorzetto}\thanksref{m1,m4}\ead[label=e1]{}},
\author{\fnms{Zizhao} \snm{Xie}\thanksref{m1, m4}\ead[label=e2]{}},
\author{\fnms{Julian} \snm{Stamp}\thanksref{m2,m3} \ead[label=e3]{}},
\\
\author{\fnms{Arman} \snm{Oganisian}\thanksref{m1}\ead[label=e4]{}} \and
\author{\fnms{Roberta} \snm{De Vito}\thanksref{m1,m2,m5}\ead[label=e5]{roberta\_devito@brown.edu}}

\runauthor{Zorzetto et al.}

\affiliation{\thanksmark{m1} Department of Biostatistics, Brown University\\ \thanksmark{m2} Data Science Institute, Brown University  \\ \thanksmark{m3} Center for Computational Molecular Biology, Brown University \\ \thanksmark{m5} Department of Statistical Science, La Sapienza University of Rome\\
\thanksmark{m4} Co-first authors.}

\iffalse

%%%%%%%%%%%%%%%%%%%%%%%%%%%%%%%%%%%%%%%%%%%%%%
%% Addresses                                %%
%%%%%%%%%%%%%%%%%%%%%%%%%%%%%%%%%%%%%%%%%%%%%%

\fi

\end{aug}

%%% ABSTRACT: max 200 words (now we have 204 words!!!!)
\begin{abstract}
\quad
Diet plays a crucial role in health, and understanding the causal effects of dietary patterns is essential for informing public health policy and personalized nutrition strategies. 
However, causal inference in nutritional epidemiology faces several challenges: (i) high-dimensional and correlated food/nutrient intake data induce massive treatment levels; (ii) nutritional studies are interested in latent dietary patterns rather than single food items; and (iii) the goal is to estimate heterogeneous causal effects of these dietary patterns on health outcomes. 
We address these challenges by introducing a sophisticated exposure mapping framework that reduces the high-dimensional treatment space via factor analysis and enables the identification of dietary patterns. Also, we extend the Bayesian Causal Forest to accommodate three ordered levels of dietary exposure, better capturing the complex structure of nutritional data, and enabling estimation of heterogeneous causal effects.

We evaluate the proposed method through extensive simulations and apply it to a multi-center epidemiological study of Hispanic/Latino adults residing in the US. Using high-dimensional dietary data, we identify six dietary patterns and estimate their causal link with two key health risk factors: body mass index and fasting insulin levels. Our findings suggest that higher consumption of plant lipid–antioxidant, plant-based, animal protein, and dairy product patterns is associated with reduced risk.
\end{abstract}


\end{frontmatter}
%%%%%%%%%%%%%%%%%%%%%%%%%%%%%%%%%%%%%%%%%%%%%%
%% Please use \tableofcontents for articles %%
%% with 50 pages and more                   %%
%%%%%%%%%%%%%%%%%%%%%%%%%%%%%%%%%%%%%%%%%%%%%%
%\tableofcontents

\section{Introduction}
Diet plays a crucial role in health outcomes, influencing the risk of chronic diseases such as cardiovascular disease, obesity, and diabetes \citep{kromhout2001diet, gbd2015global, maldonado2022posteriori, de2022shared, de2023multi}. Understanding the causal effects of dietary interventions is essential for developing effective public health policies and personalized nutrition strategies \citep{pearson2017reducing}. Although the relationship between diet and health outcomes has been extensively investigated, causal questions in nutritional epidemiology remain insufficiently addressed, in part because dietary exposures are complex and difficult to represent within standard causal inference settings.

In this paper, we bridge nutritional epidemiology and causal inference  by formalizing a causal framework tailored to nutritional epidemiological data. Several challenges arise. First, dietary intakes are typically collected as high-dimensional matrices of food or nutrient consumption, with strong correlations among items, requiring principled dimensionality reduction. Second, nutritional epidemiology commonly focuses on the identification of latent dietary patterns  that summarize this complex information and then categorize individuals into three (or more) ordered groups, i.e., low, medium, and high consumption. Third, we hypothesize that the causal relationship between dietary patterns and health outcomes exhibits heterogeneity driven by individual consumer characteristics.

Within the potential outcome framework \citep{rubin1974estimating}, many flexible methods have been developed to estimate heterogeneous causal effects \cite[see][for a review]{dominici2020controlled}. In particular, 
Bayesian nonparametric (BNP) approaches have emerged as powerful tools for causal inference, offering flexibility in modeling treatment effect heterogeneity while incorporating uncertainty quantification \citet{linero2021and}. Among these, many causal methods for heterogeneous causal effects leverage infinite mixture models \citet{roy2018bayesian, oganisian2021bayesian, zorzetto2024confounder}, Gaussian processes \citet{alaa2017bayesian}, and tree-based models \citep[BART - ][]{bargagli2020causal,athey2019generalized,deshpande2020vcbart}.
In particular, the causal formulation of Bayesian Additive Regression Trees \citep[BART - ][]{Chipman_2010} introduced by \citet{hill2011bayesian} and its extension, Bayesian Causal Forests \citep[BCF - ][]{hahn2020bayesian}, have shown strong performance due to their ability to capture complex relationships between covariates and treatment effects. BCF further improves  causal effect estimation by incorporating propensity scores and regularization techniques to mitigate confounding. Despite these advances, existing BCF models mainly target binary treatments and, more recently, discrete treatment \citep{mcjames2025bayesian}, whereas dietary interventions often involve multiple ordered treatment levels.

Despite a rich literature offering a wide range of flexible models, additional methodological development is needed to accommodate a high-dimensional treatment variable and to define an approach that reduces the dimensionality without sacrificing relevant information. Exposure mapping, introduced by \citet{savje2024causal}, formalizes the relationship between an observed high-dimensional treatment exposure and a lower-dimensional exposure variable. Related techniques have been used to address treatment spillovers and high-dimensional network confounding arising from complex unit dependencies in disease transmission networks \citep{sun2025difference}, as well as to model continuous exposures evolving over time \citep{hettinger2025multiply}. However, the highly correlated nature of nutritional data requires a more advanced exposure mapping that explicitly captures the shared latent variation across foods and nutrients.

In this paper, we make three primary methodological contributions. First, we formulate a general causal framework under exposure mapping, specifying the key properties of the mapping function, the required causal assumptions, and the estimands of interest. Second, we introduce a novel exposure mapping based on factor analysis to reduce the dimensionality of high-dimensional and highly correlated nutritional data. Third, we extend the Bayesian Causal Forest framework to accommodate three ordered treatment levels; we refer to this extension as BCF3L.

In particular, Factor analysis is widely used for dimensionality reduction, and a substantial Bayesian literature demonstrates its effectiveness to identify dietary patterns in nutritional epidemiology \citep[e.g.,][]{varraso2012assessment, de2019shared, de2023multi, huang2024sparse}. However, these approaches have largely been developed outside a causal inference framework. While Bayesian factor models have been proposed for causal inference in other applied contexts \citep{samartsidis2019assessing, samartsidis2020bayesian, zorzetto2025multivariate}, they have not been used as exposure mappings to reduce a high-dimensional treatment space in nutrition.

Our approach uses a Bayesian factor model to uncover the latent structure of observed dietary intake (the high-dimensional treatment), yielding lower-dimensional latent dietary patterns that serve as exposure variables.  Consistent with standard practice in nutritional epidemiology, these dietary patterns are discretized into ordered exposure levels, i.e., low, medium, and high consumption, such that the proposed BCF3L model can estimate heterogeneous causal effects on health outcomes for medium versus low and high versus medium exposure levels. This enables flexible, nonlinear estimation of diet–health relationships within a nonparametric Bayesian framework.

Moreover, BCF3L can be applied more broadly beyond high-dimensional settings, providing a general extension of the standard BCF framework to cases involving three ordered treatment levels.

Through extensive simulation studies, we demonstrate that our BCF3L model provides robust estimates of heterogeneous causal effects and integrates naturally with factor-analysis-based exposure mapping in high-dimensional settings. Across a range of scenarios, BCF3L outperforms both Bayesian linear regression, which imposes restrictive parametric assumptions on the treatment effect, and fully nonparametric alternatives such as BART.

We apply our method to the  Hispanic Community Health Study/Study of Latinos (HCHS/SOL), a multi-center epidemiological study designed to investigate  key factors influencing the health of Hispanic/Latino populations. We identify distinct dietary patterns and estimate their causal effects on two key cardiometabolic risk factors: body mass index (BMI) and fasting insulin levels.

We have developed the GitHub package \texttt{BCF3L} to estimate heterogeneous causal effects under three ordered treatment levels. The code that integrates exposure mapping with the Bayesian factor model and the BCF3L estimation is available on our GitHub repository: \href{https://github.com/zizhaoxie0327/Causal_nutrition}{\texttt{zizhaoxie0327/Causal\_nutrition}}.

This paper is organized as follows. Section 2 introduces the HCHS/SOL dataset and the variables included in our analysis. Section 3 presents the methodological contribution, including (i) the casual framework for exposure mapping throught the latent factors identified with the Bayesian factor model and (ii) the definition of Bayesian Causal Forests for three ordered treatment levels for a general case. Section 4 evaluates the performance of our model across various simulation settings, comparing it with the Bayesian Linear Regression model with horseshoe prior and Bayesian Additive Regression Trees. In Section 5, we estimate the causal effects of different dietary patterns on two key health risk factors: body mass index (BMI) and fasting insulin levels.

%%%%%%%%%%%%%%%%%%%%%%%%%%%%%%%%%%%%%%%%%%%%%%%%%%%%%%%%%%%%%%%%%%%%%%%%%%%%
% - section 2
%%%%%%%%%%%%%%%%%%%%%%%%%%%%%%%%%%%%%%%%%%%%%%%%%%%%%%%%%%%%%%%%%%%%%%%%%%%%

\section{HCHS/SOL Study}
\label{sub:HCHS}
\subsection{The study}
The Hispanic Community Health Study/Study of Latinos (HCHS/SOL) is the largest multi-site cohort study of Hispanics/Latinos in the United States. 
One of the key aim of the study is to provide previously unavailable information on the health status and disease burden of U.S. Hispanics/Latinos and to investigate associations between risk factors and cardiovascular disease over follow-up. The HCHS/SOL study includes 16,415 adults aged 18-74 years recruited from four different sites (i.e., Bronx, Chicago, Miami, San Diego) representing seven Hispanic/Latino backgrounds (i.e., Cuban, Dominican Republic, Mexican, Puerto Rican, Central, South American, and others). Baseline data was collected between 2008 and 2011 using a stratified two-stage probability sampling design, described in detail by  \citet{lavange2010sample}.

In this manuscript, we focus on dietary intake as a potential risk factor and assess its causal effect on health outcomes: body Mass Index(BMI)(kg/m$^2$) and fasting insulin levels (calibrated and expressed in mU/L). Fasting insulin, measured  from fasting blood samples, is a key biomarker to identify people at elevated cardiometabolic risk \citep{maldonado2022posteriori}. The BMI  was measured at the baseline visit. Fasting insulin levels were assessed at two visits; in this analysis, we use the baseline measurement, obtained from venous blood collected after a fasting period of at least 8 hours prior to the visit.

Nutritional data was collected through two 24-hour dietary recalls: the first was conducted in person at baseline and the second by telephone within 30 days after the initial visit. Interviews were conducted in the participant’s preferred language (English for 20\% and Spanish for 80\% of participants) by trained bilingual interviewers.  
Dietary intake data include 182 nutritional components such as nutrients, nutrient ratios, and food group servings. Nearly all participants (99\%) provided at least one reliable dietary recall.

\subsection{Selection of subjects and data preprocessing}

From the 16,415 baseline participants, we excluded individuals who self-identified with ethnic backgrounds outside of the seven predefined groups, had unreliable recall data (e.g., vitamin E intake < 0), or presented implausible energy intake estimates (<1st or >99th percentiles). Additionally, participants with missing data on ethnic background or annual household income were excluded. After applying these criteria, the final sample consisted of 9,298 participants. 

Table~\ref{tab:BMI_summary} summarizes socio-demographic and health-related characteristics stratified by BMI category.  Thirteen potential confounders are reported across four BMI categories: Underweight (i.e., $< 18.5$), Normal weight (i.e., $\geq 18.5$ and $< 25$), Overweight (i.e., $\geq 25$ and $< 30$), and Obese (i.e., $\geq 30$). The distribution of BMI categories varied across the four study sites, with the Obese group comprising the largest proportion of participants and the Underweight group the smallest. Overall, the majority of the participants were female, married, and reported annual household income below $30,000$. Obese individuals tended to be older, whereas underweight individuals were generally younger. Health-related characteristics indicated higher use of diabetes and hypertension medications among obese participants, while underweight individuals exhibited the highest energy intake and the lowest levels of physical activity.

Table~\ref{tab:Insulin_summary} summarizes socio-demographic and health-related characteristics stratified by fasting insulin category (i.e., Below 5, Between 5 and 15, Between 15 and 25, and Above 25 mU/L). Most participants had fasting insulin levels between 5 and 15 mU/L. Participants with higher insulin levels were generally older, had a higher prevalence of chronic conditions, and reported lower physical activity. 
Among participants with fasting insulin levels above 25 mU/L, those individuals of Mexican background comprised the largest proportion.

\begin{table*}[!h]
  \caption{\it Summary of socio-demographic and health-related characteristics across BMI (Body Mass Index, kg/m$^2$) categories.}
  \label{tab:BMI_summary}
  \tiny
  \tabcolsep=0pt%%
  \begin{tabular*}{\textwidth}{@{\extracolsep{\fill}}l c c c c@{\extracolsep{\fill}}}
    & \textbf{Underweight} & \textbf{Normal} & \textbf{Overweight} & \textbf{Obese} \\
    & ($N=39$) & ($N=1,788$) & ($N=3,618$) & ($N=3,853$) \\
    \hline \hline
    \textbf{Energy Intake (kcal)} & & & & \\
    $<$ 1300 kcal & 6 (15\%) & 352 (20\%) & 868 (24\%) & 1,056 (27\%) \\
    $\geq$ 1300 and $<$ 1700 kcal & 6 (15\%) & 429 (24\%) & 911 (25\%) & 948 (25\%) \\
    $\geq$ 1700 and $<$ 2300 kcal & 10 (26\%) & 543 (30\%) & 980 (27\%) & 1,071 (28\%) \\
    $\geq$ 2300 kcal & 17 (44\%) & 464 (26\%) & 859 (24\%) & 778 (20\%) \\
    \hline
    \textbf{Antidiabetics} & & & & \\
    No & 39 (100\%) & 1,669 (93\%) & 3,282 (91\%) & 3,223 (84\%) \\
    Yes & 0 (0\%) & 119 (6.7\%) & 336 (9.3\%) & 630 (16\%) \\
    \hline
    \textbf{Antihypertensives} & & & & \\
    No & 39 (100\%) & 1,624 (91\%) & 3,071 (85\%) & 2,929 (76\%) \\
    Yes & 0 (0\%) & 164 (9.2\%) & 547 (15\%) & 924 (24\%) \\
    \hline
    \textbf{Physical Activity} & & & & \\
    No & 14 (36\%) & 944 (53\%) & 2,143 (59\%) & 2,643 (69\%) \\
    Yes & 25 (64\%) & 844 (47\%) & 1,475 (41\%) & 1,210 (31\%) \\
    \hline
    \textbf{Well-being Score} & & & & \\
    $<$ 10 & 31 (79\%) & 1,381 (77\%) & 2,821 (78\%) & 2,804 (73\%) \\
    $\geq$ 10 and $<$ 20 & 7 (18\%) & 342 (19\%) & 664 (18\%) & 867 (23\%) \\
    $\geq$ 20 & 1 (2.6\%) & 65 (3.6\%) & 133 (3.7\%) & 182 (4.7\%) \\
    \hline
    \textbf{Employment} & & & & \\
    Retired & 1 (2.6\%) & 138 (7.7\%) & 370 (10\%) & 428 (11\%) \\
    Not retired, not currently employed & 22 (56\%) & 646 (36\%) & 1,215 (34\%) & 1,477 (38\%) \\
    Employed part-time ($\leq$ 35 hr) & 7 (18\%) & 353 (20\%) & 642 (18\%) & 601 (16\%) \\
    Employed full-time ($>$ 35 hr) & 9 (23\%) & 651 (36\%) & 1,391 (38\%) & 1,347 (35\%) \\
    \hline
    \textbf{Years in US} & & & & \\
    $<$ 10 years & 13 (33\%) & 494 (28\%) & 867 (24\%) & 710 (18\%) \\
    $\geq$ 10 and $<$ 20 years & 4 (10\%) & 408 (23\%) & 861 (24\%) & 845 (22\%) \\
    $\geq$ 20 years & 4 (10\%) & 561 (31\%) & 1,431 (40\%) & 1,651 (43\%) \\
    US Born & 18 (46\%) & 325 (18\%) & 459 (13\%) & 647 (17\%) \\
    \hline
    \textbf{Marital Status} & & & & \\
    Single & 30 (77\%) & 591 (33\%) & 773 (21\%) & 896 (23\%) \\
    Married & 8 (21\%) & 845 (47\%) & 2,090 (58\%) & 2,157 (56\%) \\
    Separated & 1 (2.6\%) & 352 (20\%) & 755 (21\%) & 800 (21\%) \\
    \hline
    \textbf{Yearly Household Income} & & & & \\
    $<$ 30k & 25 (64\%) & 1,198 (67\%) & 2,439 (67\%) & 2,635 (68\%) \\
    $\geq$ 30k & 14 (36\%) & 590 (33\%) & 1,179 (33\%) & 1,218 (32\%) \\
    \hline
    \textbf{Education Level} & & & & \\
    $<$ High School & 8 (21\%) & 567 (32\%) & 1,362 (38\%) & 1,500 (39\%) \\
    $=$ High School & 14 (36\%) & 479 (27\%) & 926 (26\%) & 959 (25\%) \\
    $>$ High School & 17 (44\%) & 742 (41\%) & 1,330 (37\%) & 1,394 (36\%) \\
    \hline
    \textbf{Gender} & & & & \\
    Female & 21 (54\%) & 1,081 (60\%) & 1,946 (54\%) & 2,430 (63\%) \\
    Male & 18 (46\%) & 707 (40\%) & 1,672 (46\%) & 1,423 (37\%) \\
    \hline
    \textbf{Participant’s Field Center} & & & & \\
    Bronx & 16 (41\%) & 403 (23\%) & 877 (24\%) & 1,041 (27\%) \\
    Chicago & 7 (18\%) & 445 (25\%) & 939 (26\%) & 1,020 (26\%) \\
    Miami & 6 (15\%) & 498 (28\%) & 893 (25\%) & 877 (23\%) \\
    San Diego & 10 (26\%) & 442 (25\%) & 909 (25\%) & 915 (24\%) \\
    \hline
    \textbf{Hispanic/Latino Background} & & & & \\
    Dominican & 7 (18\%) & 169 (9.5\%) & 360 (10.0\%) & 370 (9.6\%) \\
    Central American & 5 (13\%) & 192 (11\%) & 378 (10\%) & 377 (9.8\%) \\
    Cuban & 3 (7.7\%) & 276 (15\%) & 552 (15\%) & 538 (14\%) \\
    Mexican & 15 (38\%) & 694 (39\%) & 1,508 (42\%) & 1,555 (40\%) \\
    Puerto Rican & 8 (21\%) & 299 (17\%) & 544 (15\%) & 776 (20\%) \\
    South American & 1 (2.6\%) & 158 (8.8\%) & 276 (7.6\%) & 237 (6.2\%) \\
  \end{tabular*}
\end{table*}

Following \citet{de2022shared}, we selected 53 nutrients that best represent overall nutrition for Hispanics/Latinos. To improve identification of dietary patterns associated with cardiovascular disease, we expanded the lipid profile \citep{lichtenstein2003dietary, islam2019trans}. In particular, we constructed four additional variables: 1- short-chain saturated fatty acids (SCSFA; e.g., butyric acid), 2- medium-chain saturated fatty acids (MCSFA; sum of caproic acid, caprylic acid, capric acid, and lauric acid), 3- long-chain saturated fatty acids (LCSFA; sum of myristic acid, palmitic acid, margaric acid, stearic acid, and arachidic acid), 4- long-chain monounsaturated fatty acids (LCSFA; sum of myristoleic acid, palmitoleic acid, oleic acid, and gadoleic acid).

For each participant, nutrient intake variables were derived by averaging the two available dietary recalls for each participant.

%We log-transformed the nutrient intakes and standardized the data. 

\begin{table*}[!h]
  \caption{\it Summary of socio-demographic and health-related characteristics across insulin categories.}
  \label{tab:Insulin_summary}
  \tiny
  \tabcolsep=0pt%%
  \begin{tabular*}{\textwidth}{@{\extracolsep{\fill}}l c c c c@{\extracolsep{\fill}}}
    & \textbf{Below 5} & \textbf{Between 5 and 15} & \textbf{Between 15 and 25} & \textbf{Above 25} \\
    & ($N=1,175$) & ($N=5,429$) & ($N=1,817$) & ($N=877$) \\
    \hline \hline
    \textbf{Energy Intake (kcal)} & & & & \\
    $<$ 1300 kcal & 266 (23\%) & 1,351 (25\%) & 441 (24\%) & 224 (26\%) \\
    $\geq$ 1300 and $<$ 1700 kcal & 266 (23\%) & 1,384 (25\%) & 449 (25\%) & 195 (22\%) \\
    $\geq$ 1700 and $<$ 2300 kcal & 364 (31\%) & 1,463 (27\%) & 545 (30\%) & 232 (26\%) \\
    $\geq$ 2300 kcal & 279 (24\%) & 1,231 (23\%) & 382 (21\%) & 226 (26\%) \\
    \hline
    \textbf{Antidiabetics} & & & & \\
    No & 1,068 (91\%) & 4,931 (91\%) & 1,532 (84\%) & 682 (78\%) \\
    Yes & 107 (9.1\%) & 498 (9.2\%) & 285 (16\%) & 195 (22\%) \\
    \hline
    \textbf{Antihypertensives} & & & & \\
    No & 1,056 (90\%) & 4,598 (85\%) & 1,394 (77\%) & 615 (70\%) \\
    Yes & 119 (10\%) & 831 (15\%) & 423 (23\%) & 262 (30\%) \\
    \hline
    \textbf{Physical Activity} & & & & \\
    No & 555 (47\%) & 3,259 (60\%) & 1,276 (70\%) & 654 (75\%) \\
    Yes & 620 (53\%) & 2,170 (40\%) & 541 (30\%) & 223 (25\%) \\
    \hline
    \textbf{Well-being Score} & & & & \\
    $<$ 10 & 884 (75\%) & 4,171 (77\%) & 1,348 (74\%) & 634 (72\%) \\
    $\geq$ 10 and $<$ 20 & 245 (21\%) & 1,044 (19\%) & 385 (21\%) & 206 (23\%) \\
    $\geq$ 20 & 46 (3.9\%) & 214 (3.9\%) & 84 (4.6\%) & 37 (4.2\%) \\
    \hline
    \textbf{Employment} & & & & \\
    Retired & 88 (7.5\%) & 541 (10\%) & 209 (12\%) & 99 (11\%) \\
    Not retired and not currently employed & 375 (32\%) & 1,895 (35\%) & 703 (39\%) & 387 (44\%) \\
    Employed part-time ($\leq$ 35 hr) & 204 (17\%) & 994 (18\%) & 283 (16\%) & 122 (14\%) \\
    Employed full-time ($>$ 35 hr) & 508 (43\%) & 1,999 (37\%) & 622 (34\%) & 269 (31\%) \\
    \hline
    \textbf{Years in US} & & & & \\
    $<$ 10 years & 263 (22\%) & 1,284 (24\%) & 363 (20\%) & 174 (20\%) \\
    $\geq$ 10 and $<$ 20 years & 264 (22\%) & 1,251 (23\%) & 423 (23\%) & 180 (21\%) \\
    $\geq$ 20 years & 447 (38\%) & 2,101 (39\%) & 726 (40\%) & 373 (43\%) \\
    US Born & 201 (17\%) & 793 (15\%) & 305 (17\%) & 150 (17\%) \\
    \hline
    \textbf{Marital Status} & & & & \\
    Single & 327 (28\%) & 1,339 (25\%) & 418 (23\%) & 206 (23\%) \\
    Married & 599 (51\%) & 2,976 (55\%) & 1,023 (56\%) & 502 (57\%) \\
    Separated & 249 (21\%) & 1,114 (21\%) & 376 (21\%) & 169 (19\%) \\
    \hline
    \textbf{Yearly Household Income} & & & & \\
    $<$ 30k & 764 (65\%) & 3,648 (67\%) & 1,255 (69\%) & 630 (72\%) \\
    $\geq$ 30k & 411 (35\%) & 1,781 (33\%) & 562 (31\%) & 247 (28\%) \\
    \hline
    \textbf{Education Level} & & & & \\
    $<$ High School & 425 (36\%) & 1,980 (36\%) & 671 (37\%) & 361 (41\%) \\
    $=$ High School & 316 (27\%) & 1,374 (25\%) & 471 (26\%) & 217 (25\%) \\
    $>$ High School & 434 (37\%) & 2,075 (38\%) & 675 (37\%) & 299 (34\%) \\
    \hline
    \textbf{Gender} & & & & \\
    Female & 608 (52\%) & 3,233 (60\%) & 1,127 (62\%) & 510 (58\%) \\
    Male & 567 (48\%) & 2,196 (40\%) & 690 (38\%) & 367 (42\%) \\
    \hline
    \textbf{Participant’s Field Center} & & & & \\
    Bronx & 339 (29\%) & 1,393 (26\%) & 412 (23\%) & 193 (22\%) \\
    Chicago & 292 (25\%) & 1,389 (26\%) & 474 (26\%) & 256 (29\%) \\
    Miami & 220 (19\%) & 1,319 (24\%) & 500 (28\%) & 235 (27\%) \\
    San Diego & 324 (28\%) & 1,328 (24\%) & 431 (24\%) & 193 (22\%) \\
    \hline
    \textbf{Hispanic/Latino Background} & & & & \\
    Dominican & 135 (11\%) & 590 (11\%) & 137 (7.5\%) & 44 (5.0\%) \\
    Central American & 118 (10\%) & 517 (9.5\%) & 207 (11\%) & 110 (13\%) \\
    Cuban & 121 (10\%) & 818 (15\%) & 287 (16\%) & 143 (16\%) \\
    Mexican & 487 (41\%) & 2,199 (41\%) & 743 (41\%) & 343 (39\%) \\
    Puerto Rican & 219 (19\%) & 883 (16\%) & 330 (18\%) & 195 (22\%) \\
    South American & 95 (8.1\%) & 422 (7.8\%) & 113 (6.2\%) & 42 (4.8\%) \\
  \end{tabular*}
\end{table*}

%%%%%%%%%%%%%%%%%%%%%%%%%%%%%%%%%%%%%%%%%%%%%%%%%%%%%%%%%%%%%%%%%%%%%%%%%%%%
% - section 3
%%%%%%%%%%%%%%%%%%%%%%%%%%%%%%%%%%%%%%%%%%%%%%%%%%%%%%%%%%%%%%%%%%%%%%%%%%%%

\section{Latent treatment patterns}
\label{sec:set_up}

\subsection{Causal setup}
\label{subsec:causal_setup}
Let consider a sample of $n$ units indexed by $i \in \{1, \dots, n\}$. Let \(\mathbf{Z_i} = (Z_{i1}, \dots, Z_{ip}) \in (\mathbb{R}_0^+)^{p} \) denote the vector of $p$ observable variables representing the consumption of specific nutrient items for each individual $i$, and let ${\bf Z}\in (\mathbb{R}_0^+)^{p\times n}$ denote the collection across all $n$ units. Therefore, each unit $i$ has an infinite possible treatment assignment $\mathcal{Z} = (\mathbb{R}_0^+)^{p\times n}$, which indicates the specific dietary intake as a combination of $p$ nutrients. In our nutritional application, we have $p = 53$ nutrients and $n=9,298$ individuals.

Under Rubin Causal Model \citep{rubin1974estimating}, the potential outcome is defined as $Y_i({\bf Z}) \in \mathbb{R}$ for each unit $i$. Because $\mathcal{Z}$ is high-dimensional and continuous, the collection of potential outcomes is extremely complex. Therefore, following \citet{savje2024causal}, we introduce a function, or exposure mapping, that reduces the treatment infinite space by mapping high-dimensional intake vectors to a lower-dimensional, discrete exposure. In the following sections \ref{subsec:factors} and \ref{subsec:treatment_def}, we describe the two step procedure, respectively, factor identification and tertile reduction, to define the latent treatment ${\bf T}_i \in \mathcal{T} = \{0,1,2\}^J$. Specifically, for each unit $i$ and with $J<p$, we define 
$${\bf T}_i = t_i({\bf Z}): (\mathbb{R}_0^+)^{p\times n} \xrightarrow[]{} \{0,1,2\}^J$$. %Therefor the following  and the potential outcome is defined as $Y_i({\bf Z}) = Y_i({\bf t}_i)$, as formally defined in the following assumptions.

Hereafter, we use the term \emph{exposure} in the general setting following \citet{savje2024causal}, and \emph{dietary patterns} when referring to nutritional application.

We also observe subject-specific covariates $\mathbf{X}_i \in \mathcal{X}$, a $d$-dimension vector of background characteristics, which may be either continuous or discrete, referred to as pretreatment variables.

Given the exposure mapping \citep[following] []{savje2024causal}, we invoke the following definitions and assumptions.

\begin{definition}[Exposure mapping characteristics] \label{assump:mapping}
The exposure mapping $t(\cdot): \mathcal{Z} \xrightarrow[]{}\mathcal{T}$, with $|\mathcal{T}| < |\mathcal{Z}|$, is a surjective many-to-one function such that
\[
\forall \mathbf{t} \in \mathcal{T}, \; \exists \mathbf{z} \in \mathcal{Z}
\text{ such that } t(\mathbf{z}) = \mathbf{t},
\]
and
\[
\exists \, \mathbf{z}_1 \neq \mathbf{z}_2 \text{ such that }
t(\mathbf{z}_1) = t(\mathbf{z}_2).
\]
\end{definition}

In the nutrition application, this definition implies that multiple combinations of $p$ nutrients (${\bf z}, {\bf z}' \in \mathcal{Z}$) can map to the same dietary pattern ${\bf t}$, i.e.,   ${\bf t}= t({\bf z})= t({\bf z}')$.

\begin{definition}[Potential outcome under exposure mapping]\label{def:outcome_exp}
Given the Definition~\ref{assump:mapping} of the exposure mapping, the potential outcome under different treatment levels ${\bf z},{\bf z}'\in \mathcal{Z}$  for each unit $i \in \{1, \dots, n\}$ is defined as follows
\[
Y_i({\bf z})=Y_i({\bf z}') \text{ whenever } t_i({\bf z})=t_i({\bf z}').
\]
Equivalently, there exists a function that defines $Y_i({\bf z})=Y_i(t_i({\bf z}))$.
\end{definition}

Intuitively, different nutrients' combinations (i.e., ${\bf z},{\bf z}'$), that define the same dietary pattern $t$ (where $t=t({\bf z})=t({\bf z}')$), induce the same potential health outcome.

\begin{assumption}[No interference under exposure mapping] \label{assump:sutva_1}
Under Definitions~\ref{assump:mapping}-\ref{def:outcome_exp}, for each unit $i\in \{1, \dots, n\}$,  there is no interference between units exposure mapping, such that
\begin{equation*}
     Y_i(\mathbf{Z}) = Y_i(t_i({\bf Z}))=Y_i({\bf T}_i).
\end{equation*}
\end{assumption}

This assumption ensures that unit $i$'s potential outcome depends only on the exposure mapping of this unit. I.e., the potential health risk of each patient in the nutritional study depends on the personal dietary patterns.

\begin{assumption}[Consistency] \label{assump:cons}
For each unit $i\in \{1, \dots, n\}$, the observed outcome is the the potential outcome under the observed treatment assignment, such that 
\begin{equation*}
     Y_i = Y_i(\mathbf{z}).
\end{equation*}
Invoking Assumption~\ref{assump:sutva_1}, this can be written as
\begin{equation*}
     Y_i = Y_i(\mathbf{t}_i).
\end{equation*}
\end{assumption}

Therefore, in the nutritional study, the observed health risk of each patient $i$ is the potential health risk under the observed combinations of its nutrients intake. That is equivalent to the observed health risk under its personal dietary patterns.

\begin{assumption}[Strong ignorability] \label{assump:ignor}
Treatment assignment is random in each group of units
characterized by the covariates ${\bf X}$, such that
\begin{equation}
     \{Y_i({\bf z}): {\bf z} \in \mathcal{Z}\} \perp {\bf Z} \mid {\bf X}_i, \mbox{ for each unit  } i \in \{1, \dots, n\}.
     \label{eq:ignorability}
\end{equation}
\end{assumption}

\begin{corollary}[Strong ignorability under exposure mapping]\label{corol:ignor}
   Given the exposure mapping ${\bf T}_i = t_i({\bf Z}): (\mathbb{R}_0^+)^{p\times n} \xrightarrow[]{} \{0,1,2\}^J$ in Definition~\ref{assump:mapping} and invoking  Assumption~\ref{assump:sutva_1}, we can rewrite Eq~\eqref{eq:ignorability} of Assumption~\ref{assump:ignor} as follows
\begin{equation*}
      \{Y_i({\bf t}_i): {\bf t} \in \mathcal{T}\} \perp {\bf T} \mid {\bf X}_i, \mbox{ for each unit  } i \in \{1, \dots, n\}.
\end{equation*}
\end{corollary}

\begin{proof}
Given Assumption~\ref{assump:ignor},  Definitions~\ref{assump:mapping}-\ref{def:outcome_exp}, and the Coarea formula \citep[see][]{negro2024sample}, we can write
\begin{align*}
    f(Y({\bf t}),{\bf T}\mid {\bf X}) & = f({\bf T} \mid Y({\bf t}), {\bf X}) f(Y({\bf t})\mid {\bf X})\\
    & \overset{(\text{Coarea f.})}{=} \int_{{\bf z}:{\bf t}=t({\bf z})} \frac{f({\bf Z}\mid Y({\bf t}),{\bf X})}{\mathcal{J}_J t({\bf z})}\;d\mathcal{H}^{{p\times n} - J}({\bf z}) \; f(Y({\bf t})\mid {\bf X})\\
    & \overset{(\text{Def~\ref{def:outcome_exp}})}{=} \int_{{\bf z}:{\bf t}=t({\bf z})} \frac{f({\bf Z}\mid Y({\bf z}),{\bf X})}{\mathcal{J}_J t({\bf z})}\;d\mathcal{H}^{{p\times n} - J}({\bf z}) \; f(Y({\bf t})\mid {\bf X})\\
    & \overset{(\text{Ass~\ref{assump:ignor}})}{=} \int_{{\bf z}:{\bf t}=t({\bf z})} \frac{f({\bf Z}\mid {\bf X})}{\mathcal{J}_J t({\bf z})}\;d\mathcal{H}^{{p\times n} - J}({\bf z}) \; f(Y({\bf t})\mid {\bf X})\\
    & \overset{(\text{Coarea f.})}{=} f({\bf Z}\mid {\bf X}) \; f(Y({\bf t})\mid {\bf X});\\
\end{align*}
where $\mathcal{J}_J t({\bf z})$ denotes the $J$-dimensional Jacobian of $t(\cdot)$ and $\mathcal{H}^{{p\times n} - J}$ is the Hausdorff measure.
\end{proof}

Strong ignorability implies that, after adjusting for all relevant covariates $X$ (e.g., gender, education level, years in the US after immigration, etc), the treatment (i.e, nutrients intake) that each patient assumes is considered  random. Due to the previous assumption, this is also true for dietary patterns instead (i.e., ${\bf T}$) instead of nutrients intake  (i.e., ${\bf Z}$).

\subsection{Estimand and identifiability}
Our causal estimand of interest is the Conditional Average Treatment Effects (CATEs), where we conditional on observed covariates $\mathbf{x}$, to capture the heterogeneity that characterizes the causal effect in  our data, and the Factor Specific Conditional Average Treatment Effects (FS-CATEs) as a specific case of CATEs where we compare the causal effect on different exposure levels of one specific factor $j$, while the other ${\bf j}^{-1} =\{1,\dots, J\} \setminus j$ factors are kept fixed. This conditioning on the observed characteristics is crucial in nutritional epidemiology, where the causal effect of diet on risk factors may differ by gender, income, or level of education.

\begin{definition}[Causal estimands: CATEs and FS-CATEs]
Given two possible exposures ${\bf t}, {\bf t}' \in \mathcal{T}$, the causal estimand of CATEs is defined as:
\begin{equation}
    \tau(\mathbf{x}) := \mathbb{E}[Y_i({\bf t}) - Y_i({\bf t}') \mid \mathbf{X}_i = \mathbf{x}],
    \label{eq:CATE_general}
\end{equation}
where ${\bf t}=(t_1, \dots, t_j, \dots, t_J)$, ${\bf t}'=(t'_1, \dots, t'_j, \dots, t'_J)$, and $\exists\, j \in \{1,\dots,J\}$ such that $t_j \neq t'_j$.

The causal estimand of FS-CATE is a special case of the CATE, corresponding to the case where we want to study the conditional casual effect on varying only one factor $j$. Therefore, given the two possible exposures ${\bf t}, {\bf t}' \in \mathcal{T}$, the values for the remaining elements ${\bf j}^{-1}=\{1,\dots, J\} \setminus j$ are fixed, while we compare different values for $t_j, t_j' \in \{0,1,2\}$, such that $t_j \not= t_j'$:
\begin{equation}
    \tau_{(t_j',t_j)}(\mathbf{x},{\bf t}_{j^{-1}}) := \mathbb{E}[Y_i(T_{ij}=t_j, {\bf T}_{ij^{-1}}={\bf t}_{j^{-1}}) - Y_i(T_{ij}=t'_j, {\bf T}_{ij^{-1}}={\bf t}_{j^{-1}}) \mid \mathbf{X}_i = \mathbf{x}]. \label{eq:fs-cates}
\end{equation}
\end{definition}

FS-CATEs are particularly important in nutrition because they quantify how changing adherence to a single dietary pattern affects the outcome while holding other patterns fixed.
For example, in Figure~\ref{fig:outcome}, $\tau_{0,1}$ for the plant lipid–antioxidant factor and BMI represents the causal effect (FS-CATE) of moving from low to medium exposure; a negative value indicates a reduction in BMI. In particular, the negative value indicates that increasing the consumption of plant lipid–antioxidants from low to medium, while keeping the other nutritional patterns fixed, is associated with a reduction in BMI.

Both estimands are identifiable under Corollary~\ref{corol:ignor} and Assumption~\ref{assump:cons}.
Specifically, we have the following:
\begin{align}
    \tau(\mathbf{x}) & = \mathbb{E}[Y_i({\bf t})\mid \mathbf{X}_i = \mathbf{x}] - \mathbb{E}[Y_i({\bf t}') \mid \mathbf{X}_i = \mathbf{x}]\notag\\
    & \overset{(\text{Corol~\ref{corol:ignor}})}{=} \mathbb{E}[Y_i({\bf t})\mid \mathbf{X}_i = \mathbf{x}, {\bf T}_i = {\bf t}] - \mathbb{E}[Y_i({\bf t}') \mid \mathbf{X}_i = \mathbf{x}, , {\bf T}_i = {\bf t}']\notag\\
    & \overset{(\text{Ass~\ref{assump:cons}})}{=} \mathbb{E}[Y_i\mid \mathbf{X}_i = \mathbf{x}, {\bf T}_i = {\bf t}] - \mathbb{E}[Y_i \mid \mathbf{X}_i = \mathbf{x}, , {\bf T}_i = {\bf t}'];\notag\\
    \tau^j_{(t'_j,t_j)}(\mathbf{x},{\bf t}_{j^{-1}}) & = \mathbb{E}[Y_i(T_{ij}=t_j, {\bf T}_{ij^{-1}}={\bf t}_{j^{-1}}) \mid \mathbf{X}_i = \mathbf{x}] - \mathbb{E}[Y_i(T_{ij}=t'_j, {\bf T}_{ij^{-1}}={\bf t}_{j^{-1}}) \mid \mathbf{X}_i = \mathbf{x}]\notag\\
    & \overset{(\text{Corol~\ref{corol:ignor}})}{=} \mathbb{E}[Y_i(T_{ij}=t_j, {\bf T}_{ij^{-1}}={\bf t}_{j^{-1}}) \mid \mathbf{X}_i = \mathbf{x}, T_{ij}=t_j, {\bf T}_{ij^{-1}}={\bf t}_{j^{-1}}]\notag \\
    & \quad\quad\, - \mathbb{E}[Y_i(T_{ij}=t'_j, {\bf T}_{ij^{-1}}={\bf t}_{j^{-1}}) \mid \mathbf{X}_i = \mathbf{x}, T_{ij}=t'_j, {\bf T}_{ij^{-1}}={\bf t}_{j^{-1}}]\notag\\
     & \overset{(\text{Ass~\ref{assump:cons}})}{=} \mathbb{E}[Y_i \mid \mathbf{X}_i = \mathbf{x}, T_{ij}=t_j, {\bf T}_{ij^{-1}}={\bf t}_{j^{-1}}] - \mathbb{E}[Y_i \mid \mathbf{X}_i = \mathbf{x}, T_{ij}=t'_j, {\bf T}_{ij^{-1}}={\bf t}_{j^{-1}}].\label{eq:fs_identif}
\end{align}
Conditional expectations $\mathbb{E}[Y_i \mid \cdot]$ are functions of observable variables, so they are identifiable. In this work, we estimate them using the Bayesian Causal Forest with three ordered treatment levels introduced in the following section.

\subsection{Factor analysis}
\label{subsec:factors}
The first step in constructing the exposure mapping is factor analysis to identify the factors $J < p$, where $J$ and $p$ are the dimensions of the spaces $\mathcal{T}$ and $\mathcal{Z}$, respectively.

Factor analysis is a widely used method for extracting latent constructs \citep[see][]{Lopes2004, bhattacharya2011} and has become a well-established approach in nutritional epidemiology for identifying dietary patterns \citep[e.g.][]{edefonti2012nutrient, de2019shared, castello2022adherence, de2022shared}.
Factor analysis is a statistical method that models observed high-dimensional data through a lower-dimensional latent structure, aiming to identify underlying patterns in complex datasets \citep{thurstone1931multiple, anderson2004introduction}. Let \(\mathbf{Z_i} \in (\mathbb{R}_0^+)^{p} \) the data vector where each component represents the consumption of a specific food item for unit $i$, its standardization is defined as $Z^*_{ip}=(Z_{ip}- \mu_p)/\sigma_p$, where $\mu_p$ and $\sigma_p$ are the mean and standard deviation of the $p$-th variable, such that $Z^*_{ip} \in \mathbb{R}$.

Factor analysis assumes the following model:
\begin{equation}
    \mathbf{Z}^*_i = \Lambda \mathbf{l}_i + 
\boldsymbol{\varepsilon}_i,
\end{equation}
where \(\Lambda\) is a \(p \times J\) factor loading matrix, \(\mathbf{l}_i\) is a \(J\)-dimensional vector of latent factor scores, and $J$ is the number of factors. The error term \(\boldsymbol{\varepsilon}_i\) is assumed to be normally distributed, $\boldsymbol{\varepsilon}_i \sim \mathcal{N}_P(0, \Psi)$, with $\Psi \in \mathbb{R}^{p \times p}$ is diagonal matrix with elements $(\psi_{1}, \dots, \psi_p)$.

The latent factor scores are mutually independent and normally distributed, i.e,  $\mathbf{l}_i \sim \mathcal{N}_J(0, I_J)$, where $I_J$ is the $J \times J$ identity matrix.
Consequently, the covariance of the observed data can be expressed as: $\Sigma = \Lambda \Lambda^T + \Psi,$ where the off-diagonal elements of \(\Lambda \Lambda^T\) capture the pairwise covariances of the $p$ variables.

In the Bayesian factor analysis framework, priors are imposed on the factor loading matrix to encourage shrinkage and sparsity, and facilitate estimation. Among the priors proposed in the literature, we adopt the  gamma process shrinkage prior, introduced by \citet{bhattacharya2011}, applied to each element \(\lambda_{hj}\) of the factor loading matrix $\Lambda$, where $h\in \{1 \ldots p\}$ indexes the observed variables and $j\in \{1, \ldots, \infty\}$  indexes the factor:

\begin{equation*}
    \lambda_{hj} | \omega_{hj}, \tau_j \sim \mathcal{N}(0, \omega_{hj}^{-1} \tau_j^{-1}), 
\end{equation*}
where \(\omega_{hj} \sim \Gamma(\nu/2, \nu/2)\) represents the local shrinkage and \(\tau_j = \prod_{l=1}^{j} \delta_l\) controls global shrinkage across factors. The independent stochastic process \(\{\delta_l\}_{l\leq 1}\) follows Gamma probability distributions, such that the fist element $l=1$ has distribution equal to $\delta_1 \sim \Gamma(a_1, 1)$, while the following elements $l\leq 2$ has distribution
\(\delta_l \sim \Gamma(a_2, 1)\).This prior structure allows an automatic selection of the number of factors, effectively truncating it to an upper bound \(J^*\), usually smaller than \(p\).

To conclude the prior specification, we follow the well-established literature \citep[see][]{Lopes2004, rovckova2016}, where the diagonal elements of the error variance matrix \(\mathbf{\Psi}\) have an inverse-gamma distribution as prior: $\psi_p^{-2} \sim \Gamma(a^\psi, b^\psi)$.

\subsection{Exposure definition}
\label{subsec:treatment_def}
To complete the definition of the exposure mapping $t_i({\bf z}):  (\mathbb{R}_0^+)^{p\times n} \xrightarrow[]{} \{0,1,2\}^J$, where $J$ are the factors identified in Section \ref{subsec:factors}, we map the posterior point estimation of the factor score $l_{ij}$, indicated with $\hat{l}$, in a discrete ordered set such that $\hat{l}_{ij} \in \mathbb{R} \xrightarrow[]{} \{0,1,2\}$, for each unit $i$ and factor $j$.

Specifically, we define three ordered levels for the variable $T_{ij}$ by partitioning individuals for each factor $j\in \{1, \dots, J\}$ into three groups based on the terciles of the corresponding factor scores $\hat{\bf l}_j=\hat{l}_{1j}, \dots \hat{l}_{nj}$. The use of tertiles is driven by the common choice in nutritional epidemiology  to define low, medium, and high adherence/consumption groups \citep[e.g.][]{edefonti2012nutrient, de2019shared, castello2022adherence, de2022shared}.

For each factor \( j \), let \( Q_1^j \) and \( Q_2^j \) denote the first and second terciles of the distribution of $\hat{\bf l}_j$, respectively. We then assign each individual \( i \) to one of three treatment levels based on these cutoffs:

\begin{equation}
t_{ij} =
\begin{cases}
0, & \text{if } \hat{l}_{ij} \leq Q_1^j, \\
1, & \text{if } Q_1^j < \hat{l}_{ij} \leq Q_2^j, \\
2, & \text{if } \hat{l}_{ij} > Q_2^j.
\end{cases}
\label{eq:treatment_def}
\end{equation}

Thus, for each unit $i\in \{1, \dots, n\}$, the function $t_{i}(\mathbf{Z}): (\mathbb{R}_0^+)^{p\times n} \xrightarrow[]{} \{0,1,2\}^J$, defines $J$ independent factors (where the independency come from the properties of factor analysis). We indicates \( T_{ij} = 0 \) the low consumption of dietary pattern $j$,  \( T_{ij} = 1 \) the medium consumption of the dietary pattern $j$, and \( T_{ij} = 2 \) the high consumption of the same pattern.

\subsection{Bayesian causal forest for three treatment levels}
In this paper, we propose a flexible Bayesian nonparametric model to estimate $\mathbb{E}[Y_i \mid \mathbf{X}_i = \mathbf{x}, T_{ij}=t, {\bf T}_{ij^{-1}}={\bf t}_{j^{-1}}]$ in Equation~\ref{eq:fs_identif}, leveraging the Bayesian Additive Regression Tree (BART) introduced by \citet{Chipman_2010}.
Specifically, we extend the Bayesian Causal Forest (BCF) model, introduced by \citet{hahn2020bayesian} in the context of binary treatment, to handle the setting with three ordered treatment levels (BCF3L).

Formally, we assume that the outcome variable is modeled as follows:
\begin{gather}
Y_i \mid \mathbf{X}_i = \mathbf{x}, T_{ij}=t_j, {\bf T}_{ij^{-1}}={\bf t}_{j^{-1}} \sim f(\mathbf{x},  t_j, {\bf t}_{j^{-1}}) + \epsilon_i, \label{eq:model} \\
f(\mathbf{x},  t_j, {\bf t}_{j^{-1}}) = \mu(\mathbf{x}, {\bf t}_{j^{-1}}, \boldsymbol{\hat{\pi}}_{ij}) + \tau_{0,1}(\mathbf{x}, {\bf t}_{j^{-1}})\mathbb{I}_{\{t_{j} > 0\}} + \tau_{1,2}(\mathbf{x}, {\bf t}_{j^{-1}})\mathbb{I}_{\{t_{j} > 1\}},\notag\\
\epsilon \sim N( 0, \sigma^2); \notag
\end{gather}
where $\mu$, $\tau_{0,1}$, and $\tau_{1,2}$ are given independent BART priors \citep{Chipman_2010}.
Specifically, each function is modeled as a sum of 
$b$ regression trees:
\begin{gather*}
\mu(\cdot) = \sum_{j=1}^{b_\mu} g_\mu(\cdot; A_{\mu, j},M_{\mu, j}) \quad \mbox{ and }\quad \tau_{t,t'}(\cdot) = \sum_{j=1}^{b_\mu} g_{\tau_{t,t'}}(\cdot; A_{\tau_{t,t'}, j},M_{\tau_{t,t'}, j}),
\end{gather*}
where $\{A_{\cdot, j}\}_{j=1}^{b_\cdot}$ denote the sequence of binary trees that are composed of decision rules that lead to bottom nodes, also called leaves. Given the $m$ bottom nodes for each $j$-th tree, the parameters $M_{\mu,j}=\{\mu_{j1},\dots, \mu_{jm}\}$ are the mean responses of the subgroups of observations that fall in each node.

The independent BART prior allows (i)  the definition of a flexible function to capture the treatment effect heterogeneity through the sum of many different trees, and (ii) the use of shrinkage parameters to maintain regularization and prevent overfitting. In particular, regularization is achieved by controlling the probability of a node splitting through the prior  $\eta(1 + h)^{-\beta}$, where $h$ denotes the tree depth. The parameters $\eta\in (0,1)$ and $\beta \in [0,\infty)$ define a decreasing function that ensures that deeper trees are less probable. The leaf parameters are drawn from a normal distribution with variance scaled by the number of trees, encouraging smooth function estimation. 

Following \citet{hahn2020bayesian}, $\mu(\cdot)$ models the baseline response, while $\tau_{1,0}(\cdot)$ and $\tau_{2,1}(\cdot)$ model the heterogeneous treatment effects in increasing treatment levels, such that:
\begin{gather*}
\mu(\mathbf{x}, {\bf t}_{j^{-1}}, \boldsymbol{\hat{\pi}}_{ij}) = \mathbb{E}[Y_i \mid T_{ij}=0, {\bf T}_{ij^{-1}}={\bf t}_{j^{-1}}, \mathbf{X}_i=\mathbf{x}, \boldsymbol{\hat{\pi}}_{ij}] \\
\tau_{0,1}(\mathbf{x},{\bf t}_{j^{-1}}) = \mathbb{E}[Y_i \mid T_{ij}=1, {\bf T}_{ij^{-1}}={\bf t}_{j^{-1}}, \mathbf{X}_i=\mathbf{x}] - \mathbb{E}[Y_i \mid T_{ij}=0, {\bf T}_{ij^{-1}}={\bf t}_{j^{-1}}, \mathbf{X}_i=\mathbf{x}],\\
\tau_{1,2}(\mathbf{x},{\bf t}_{j^{-1}}) = \mathbb{E}[Y_i \mid T_{ij}=2, {\bf T}_{ij^{-1}}={\bf t}_{j^{-1}}, \mathbf{X}_i=\mathbf{x}] - \mathbb{E}[Y_i \mid T_{ij}=1, {\bf T}_{ij^{-1}}={\bf t}_{j^{-1}}, \mathbf{X}_i=\mathbf{x}].
\end{gather*}
where $\tau_{t'_j,t_j}(\mathbf{x})$ represents the statistical estimands for the FS-CATEs defined in Eq.~\eqref{eq:fs_identif}. 

Into the baseline response $\mu(\cdot)$, we incorporate the estimated generalized propensity scores for the multi-level treatment, denoted as $\boldsymbol{\hat{\pi}}_{ij}$. Specifically, generalized propensity scores are defined as $\boldsymbol{\hat{\pi}}_{ij} =( \pi_{ij1}, \pi_{ij2})$, where $ \pi_{ijt} = \mathbb{P}(T_{ij} = t|\mathbf{X}_i={\bf x}, {\bf T}_{ij^{-1}}={\bf t}_{j^{-1}})$ indicates the probability of observed that level $t \in \{1,2\}$ for the $j$-th factor given the confounder ${\bf x}$ and the other $j^{-1}$ factors equal to ${\bf t}_{j^{-1}}$.

%%%%%%%%%%%%%%%%%%%%%%%%%%%%%%%%%%%%%%%%%%%%%%%%%%%%%%%%%%%%%%%%%%%%%%%%%%%%
% - section 4
%%%%%%%%%%%%%%%%%%%%%%%%%%%%%%%%%%%%%%%%%%%%%%%%%%%%%%%%%%%%%%%%%%%%%%%%%%%%

\section{Simulation results}
\label{sec:simulation}

We conducted an extensive simulation study to evaluate the proposed approach: the BCF3L model and the exposure mapping based on the factor model.
In particular, we assess the ability to recover heterogeneity in the two causal effects $\tau_{0,1}$ and $\tau_{1,2}$, defined in Eq. \eqref{eq:fs-cates}.
We compare our BCF3L  with Bayesian linear regression using the horseshoe prior (BLR-HR) by \citet{carvalho2010horseshoe} and the BART model for three ordered treatment levels. %proposed by \citet{Yifan2025}.

We consider three data-generating scenarios. Scenarios~1 and~2 involve a single observed treatment variable with three ordered levels and are designed to evaluate the performance of BCF3L in settings where exposure mapping is not required. Scenario~3 mirrors the nutritional application, featuring a high-dimensional treatment variable for which exposure mapping via a Bayesian factor model is essential to identify latent factors that serve as exposure variables.
For each scenario, we simulate $50$ replicates, with a sample size of $n=500$.

Across all scenarios, we generate five covariates $\mathbf{X}=(X_1,\ldots,X_5)$. The first three are continuous: $X_1 \sim \mathcal{N}(0,1)$, $X_2 \sim \mathcal{N}(0,2^2)$, and $X_3 \sim \mathcal{N}(0,3^2)$. The remaining two are Bernoulli variables with success probabilities 0.7 and 0.5, respectively.  In Scenarios~1 and~2, the treatment variable $T \in \{0,1,2\}$ is assigned from a multinomial distribution with probabilities proportional to $ X_1 + X_2$, $0.4X_2-X_1$, and $1-0.1X_3+X_4$. 
In Scenario~3, treatment assignment is derived from latent factors as described below.

\vspace{0.25em}
\underline{Scenario 1:} The linear baseline response function and the heterogeneous treatment effects are, respectively,
\begin{gather*}
    \mu({\bf X}) = 2 - e^{X_1/3} + 1.5X_4X_5 - 2(1-X_4)(1-X_5),\\
    \tau({\bf X})_{0,1} = 0.2X_4,\\
    \tau({\bf X})_{1,2} = 1 + 0.2X_2X_5 + 0.7X_1X_4.
\end{gather*}

\vspace{0.25em}
\underline{Scenario 2:} The nonlinear baseline response function and the heterogeneous treatment effects are, respectively,
\begin{gather*}
    \mu({\bf X}) = 2 - \mid X_1/3-1\mid + 1.5X_4X_5 - 2(1-X_4)(1-X_5),\\
    \tau({\bf X})_{0,1} = 0.1X_2,\\
    \tau({\bf X})_{1,2} = 1 + 0.5 \mid X_3\mid + 0.7X_2X_4.
\end{gather*}

For Scenarios~1 and~2, the outcome is generated according to the Equation \eqref{eq:model}.

\vspace{0.25em}
\underline{Scenario 3:} The variable $Z$ is a matrix with $n=500$ units and $p=20$ variables simulated from a factor model with $J=4$ factors. The factor loading matrix $\Lambda$ of dimension $20 \times 4$ has entries drawn from a uniform distribution on $\mathrm{Unif}[0.6,1]$, half of the entries are randomly sign-flipped to induce variability and one-fifth are set to zero to induce sparsity. Factor scores $\mathbf{l}$ are a matrix $4\times 500$ with values from a standard Gaussian distribution. Each factor ($j=1,2,3$) is discretized into tertiles according to Eq.~\eqref{eq:treatment_def}, yielding three ordered levels for treatment variable $T_j$ ($j=1,2,3$). The outcome depends on both the treatment variable and a baseline response function that is a constant:

%depend on the covariates
\begin{gather*}
    Y = 2 +0.5X_1\mathbb{I}_{\{T_1\geq1\}} +0.1X_2\mathbb{I}_{\{T_1=2\}} +0.2X_1\mathbb{I}_{\{T_2\geq1\}} + 0.3X_3\mathbb{I}_{\{T_2=2\}} \\ \; - 0.2X_3\mathbb{I}_{\{T_3\geq1\}} -X_2/3\mathbb{I}_{\{T_3=2\}} + 0.4X_1\mathbb{I}_{\{T_4=2\}} +0.2X_2\mathbb{I}_{\{T_4\geq1\}}.
\end{gather*}

Figure~\ref{fig:sim_bias} and Figure~\ref{fig:sim_rmse} summarize  bias and root mean squared error (RMSE), respectively.

In Scenario 1, BCF3L achieves bias centered around zero for both causal contrasts, $\tau_{0,1}$, and $\tau_{1,2}$, closely matching BART, whereas BLR-HS exhibits evident bias, particularly for $\tau_{0,1}$. The RMSE results confirm that BCF3L  yields the smallest overall error. 

In Scenario 2, BCF3L  and BART perform comparably in terms of bias and RMSE for $\tau_{1,2}$, but BCF3L  estimates $\tau_{0,1}$ with smaller RMSE, indicating more accurate point-estimation compared to BART. In contrast, BLR-HS shows substantial bias and inflated RMSE for both causal effects.

In Scenario 3, which requires exposure mapping, BCF3L  still provides the lowest median RMSE, although a few outliers appear in the bias distribution. BART exhibits wider interquartile ranges in bias, and BLR-HS remains systematically biased, especially for $\tau_{0,1}$. Overall, BCF3L demonstrates superior accuracy and uncertainty calibration across all three scenarios.

Additional simulation results, including  coverage rates, are reported in the Supplementary Material (Section: \emph{More results of simulation study}). These results show that BCF3L  shows consistently high coverage while preserving low bias and RMSE.
Additional simulation results, including coverage rates, are reported in the Supplementary Material (Section~A). These results indicate that BCF3L attains consistently high coverage while maintaining low bias and RMSE.

\begin{figure}[ht] 
 \centering
    \includegraphics[width=0.9\textwidth]{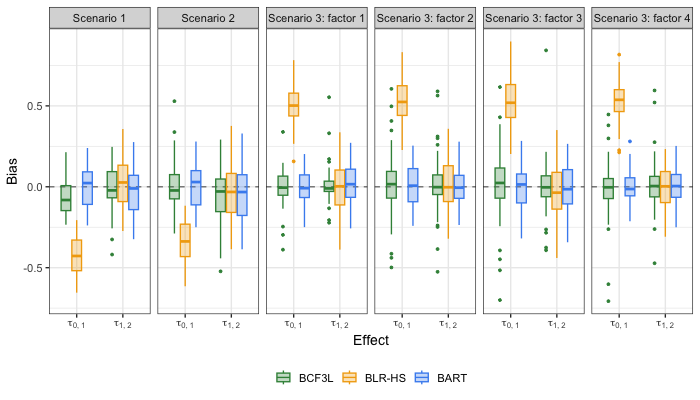} 
    \caption{Bias of the estimated causal effects $\tau_{0,1}$ (medium vs.\ low) and $\tau_{1,2}$ (high vs.\ medium)  across 50 simulated datasets.  Panels (left to right) show: Scenario~1, Scenario~2, and Scenario 3 (factors 1-4). Colors denote methods: BCF3L (green), BLR‑HS (yellow), and BART (blue).}
    \label{fig:sim_bias}
\end{figure}

\begin{figure}[ht] 
 \centering
    \includegraphics[width=0.9\textwidth]{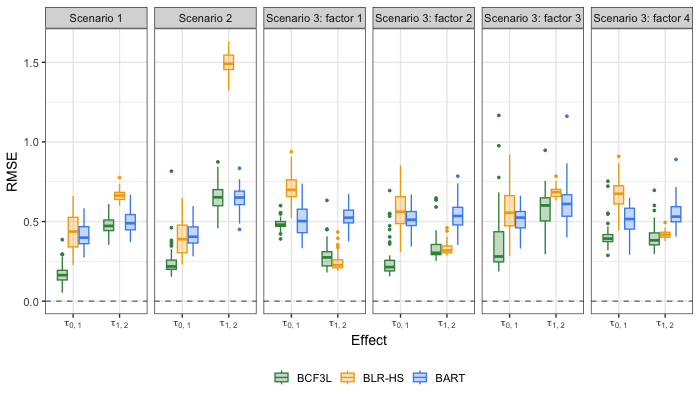} 
    \caption{Root mean squared error (RMSE) of the estimated causal effects $\tau_{0,1}$ (medium vs.\ low) and $\tau_{1,2}$ (high vs.\ medium). Panels (left to right) correspond to Scenario~1, Scenario~2, Scenario~3 (factors 1–4).  Colors denote methods: BCF3L (green), BLR‑HS (yellow), and BART (blue).}
    \label{fig:sim_rmse}
\end{figure}

%%%%%%%%%%%%%%%%%%%%%%%%%%%%%%%%%%%%%%%%%%%%%%%%%%%%%%%%%%%%%%%%%%%%%%%%%%%%
% - section 5
%%%%%%%%%%%%%%%%%%%%%%%%%%%%%%%%%%%%%%%%%%%%%%%%%%%%%%%%%%%%%%%%%%%%%%%%%%%%

\section{Nutritional Epidemiology Application}
\label{sec:application}

To investigate the causal relationship between dietary intake and key health risk factors, we apply our methodology to the HCHS/SOL study, as described in Section \ref{sub:HCHS}. We first identify dietary patterns through the proposed exposure mapping. As defined in Section~\ref{sec:set_up}, we first estimate the factor model such that, for each pattern, we obtain the individual-level factor score (pattern adherence) and then discretize it into three ordered exposure levels (low, medium, high) yielding a three-level ordered treatment $T_j \in \{0,1,2\}$, as in Eq.~\eqref{eq:treatment_def}. Finally, we estimate the heterogeneous causal effects of these dietary exposures on two outcomes: the body mass index (BMI) and the fasting insulin levels (mU/L).

Following the notation in Section \ref{subsec:factors}, $\mathbf{Z}$ denotes the nutritional data matrix of dimension $9,298 \times  53$ (participants $\times$ nutrients) after the preprocessing described in Section \ref{sub:HCHS}. We analyze two continuous outcomes $Y$, BMI and fasting insulin levels, summarized in Tables \ref{tab:BMI_summary} and \ref{tab:Insulin_summary}.

\subsection{Dietary Patterns}
We first estimate the Bayesian factor model and select the number of factors (i.e.,  dietary patterns) using the explained variance, a standard criterion balancing dimensionality reduction and variance explained. Specifically,  we retain factors with explained variance $>5\%$, yielding $k=6$ patterns. We then refit the model with $k=6$  to obtain the factor loading and score.

Figure~\ref{fig:heatmap} presents the heatmap of the estimated factor loading matrix, after applying the varimax transformation \citep{kaiser1960application}. Following nutritional literature, we name each factor based on important loadings (i.e.\(|\lambda_{hj}| \geq 0.45\)). The first factor, namely \textit{plant lipid–antioxidant}, shows significant loadings on starch, linoleic acid, linolenic acid, LCSFA (long-chain saturated fatty acids) and LCMFA (long-chain monounsaturated fatty acids), a nutrient profile commonly found in plant-derived oils and foods rich in antioxidants, such as seeds and whole grains. The second factor, labeled \textit{dairy product}, is marked by high coefficients for calcium, vitamin D, vitamin B2 and SCSFA (short-chain saturated fatty acids), which are distinctive nutrients of milk and dairy-based foods. The third factor, representing \textit{processed food}, loads heavily on trans fats and industrial fatty acids (linolelaidic and elaidic acids): compounds not naturally abundant in whole foods but introduced through industrial processing. The fourth factor, called \textit{plant-based}, clusters vegetable protein, natural folate, phytic acid, magnesium, both soluble and insoluble dietary fiber, reflecting whole-plant nutrient patterns, predominantly sourced from legumes, whole grains, and vegetables. The fifth component, indicating \textit{animal protein}, is defined by animal protein content, methionine, selenium, and arachidonic acid; a micronutrient and amino acid profile characteristic of meat, eggs, and other animal-sourced foods. While the sixth factor, namely \textit{seafood}, captures marine lipid signatures with dominant loadings on EPA (eicosapentaenoic acid), DPA (docosapentaenoic acid), DHA (docosahexaenoic acid), which are long-chain omega-3 fatty acids found primarily in fish and other seafood.

Each of these six dietary patterns defines an exposure variable $T_j$, with $j=1, \dots, 6$. As in Eq.~ \eqref{eq:treatment_def}, we divide each factor score distribution into tertiles, assigning individuals to one of three ordered treatment levels: low, medium, or high consumption.
Consequently,  an individual can simultaneously fall into different exposure levels across patterns (e.g., a plant-forward eater may be high on plant-based and low on animal protein or seafood).

\begin{figure}[ht] 
    \centering
    \includegraphics[width=0.9\textwidth]{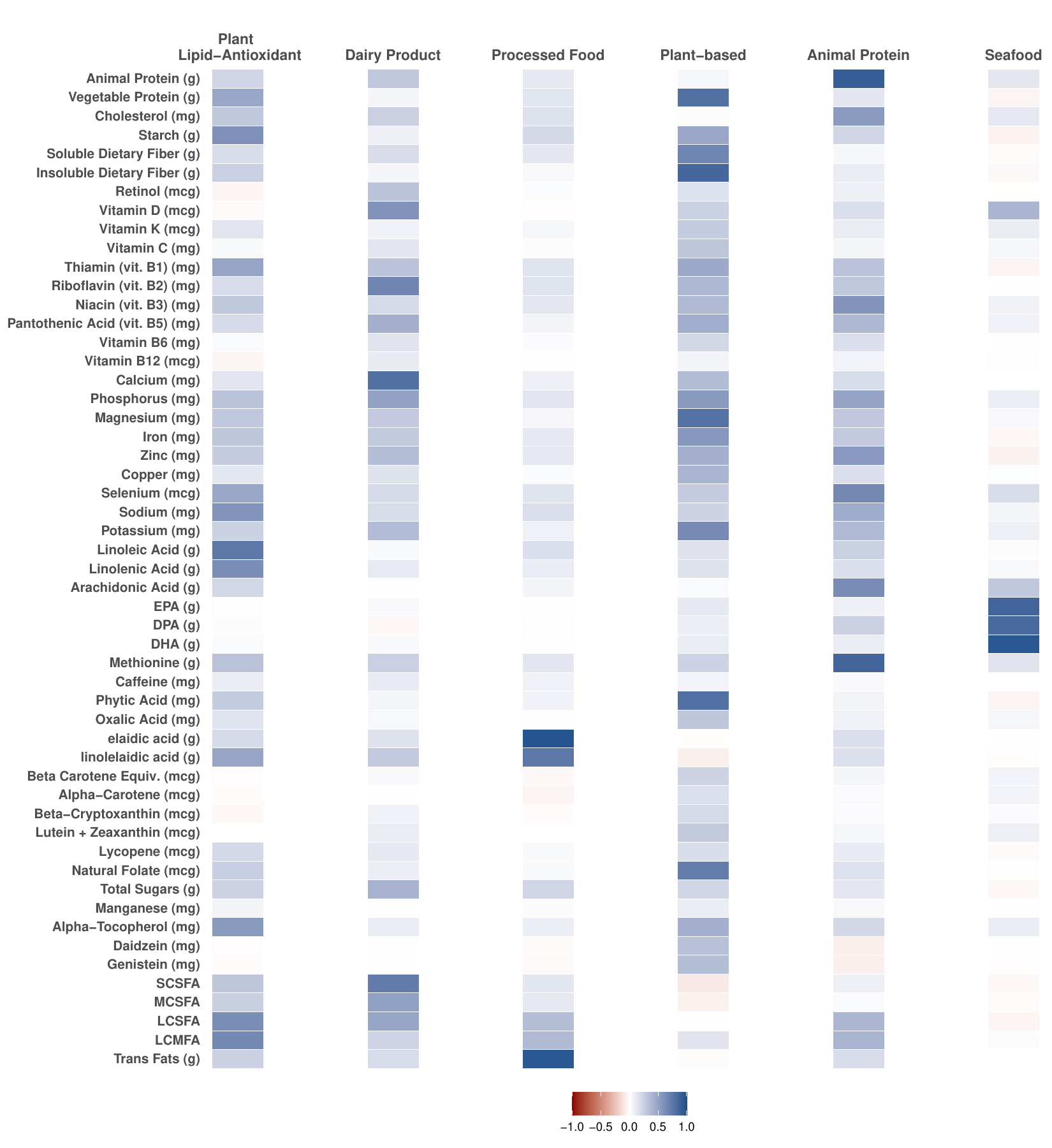} % Adjust width as needed
    \caption{Varimax-rotated factor loadings for 53 nutrients in HCHS/SOL across the six dietary patterns.}
    \label{fig:heatmap}
\end{figure}

\subsection{Heterogeneous Causal Effects}
Then, we estimate the causal effect of each dietary pattern on BMI and fasting insulin levels by fitting our BCF3L model separately for each outcome, adjusting for the covariates described in Section~\ref{sub:HCHS}.

Figure~\ref{fig:outcome} presents the  posterior means and the corresponding 90\% credible intervals for the causal effects $\tau_{0,1}$ (medium vs.\ low) and $\tau_{1,2}$ (high vs.\ medium).  The estimated causal effects of increasing food intake from low to medium levels, compared to increases from medium to high, exhibit wider credible intervals across dietary patterns and health outcomes. This affects the statistical significance of the causal effects, as indicated by whether the credible intervals include zero.For example, for BMI under the \textit{plant-based} pattern, the posterior means of $\tau_{0,1}$ and $\tau_{1,2}$ are both close to zero; however, the credible interval for $\tau_{0,1}$ excludes zero (suggesting a modest reduction in BMI with moderate intake), whereas the interval for $\tau_{1,2}$ includes zero (no significant effect from medium to high). Similar behavior is observed for the plant lipid-antioxidant and seafood patterns. In general, these results suggest that high increases in consumption may not substantially affect health outcomes, while slight increases in the consumption of specific dietary patterns can lead to moderate changes in BMI and insulin regulation.

In more  detail, higher adherence to the plant-lipid antioxidant pattern is significantly associated with lower BMI, whereas effects on fasting insulin are not statistically significant at either consumption type. This inverse association with BMI is plausible given higher intake of complex carbohydrates and antioxidant-rich lipids, which can enhance satiety and reduce oxidative stress (\citet{Nanri2013}; \citet{Aune2018}; \citet{Zhao2020}).

Similarly, higher adherence to the dairy product pattern is negatively associated with BMI and fasting insulin levels. Similar associations between vitamin D, calcium intake and both BMI and insulin resistance have been reported in population-based studies (\citet{Pittas2010}; \citet{Song2013}; \citet{Guo2020}). 

No statistically significant effects are detected for the processed food pattern on BMI or fasting insulin, although a causal positive trend for fasting insulin is evident. This direction aligns with evidence linking trans-fat intake to insulin resistance (\citet{Mozaffarian2006}; \citet{Wang2021}). 

Higher adherence to the plant-based pattern isassociated with lower BMI (not always significant at the highest intake) and with reduced fasting insulin levels among individuals with medium to high consumption. These findings suggest that plant-derived nutrients, dietary fiber and natural folate can improve insulin sensitivity through enhanced glucose regulation, reduced oxidative stress, and modulation of gut–metabolic interactions. These findings are consistent with benefits of fiber and folate for glycemic regulation and insulin sensitivity(\citet{McKeown2018}; \citet{Kim2021}).

In contrast, the Animal-Protein Pattern presents a slightly different conclusion. For BMI, the contrast from low to medium exposure ($\tau_{0,1}$) is associated with a larger reduction than the contrast from medium to high ($\tau_{1,2}$), which is closer to zero.  For fasting insulin, however, the magnitude of $\tau_{1,2}$ is larger, suggesting that higher intake of animal-based foods leads to a more pronounced improvement in insulin regulation. Although such favorable metabolic effects of animal-protein–rich diets are not commonly reported, several mechanisms may explain this trend. Moderate protein intake can enhance satiety, preserve lean body mass, and increase thermogenic energy expenditure, while zinc and selenium support pancreatic $\beta$-cell function and antioxidant defense \citep{Weigle2005}. Moreover, methionine contributes to methylation and glutathione synthesis, potentially improving redox balance and insulin sensitivity (\citet{VanNielen2014}; \citet{Rietman2014}).

Finally, a higher adherence to the seafood pattern is associated with a reduction in both BMI and fasting insulin levels. The estimated causal effects  are consistent with biological evidence that long-chain n-3 polyunsaturated fatty acids (EPA/DHA)  can improve membrane fluidity, modulate inflammation, and enhance insulin sensitivity and promote healthier body composition  (\citet{Virtanen2014}; \citet{Lankinen2018}; \citet{Belalcazar2021}).

\begin{figure}[ht] 
    \centering
    \includegraphics[width=0.9\textwidth]{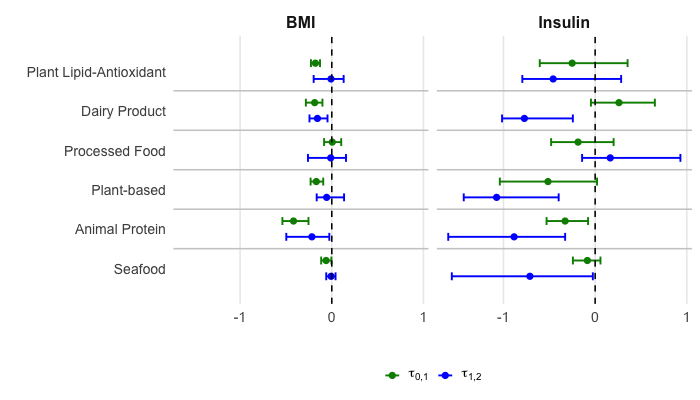} % Adjust width as needed
    \caption{Posterior means and 90\% credible intervals of the causal effects of the six identified dietary patterns on body mass index (BMI) and fasting insulin levels: $\tau_{0,1}$ (medium vs.\ low) and $\tau_{1,2}$ (high vs.\ medium).  Estimates obtained from BCF3L adjusting for the covariates in Section~\ref{sub:HCHS}}
    \label{fig:outcome}
\end{figure}

%This case study illustrates the need of the Bayesian factor analysis model with three ordered treatment levels in observational epidemiology. By deriving ordered exposures from a factor model, our approach recovers the distinct effects of initial intake versus further increases and provides interpretable contrast estimates. The model flexibly adjusts for confounding and naturally accommodates heterogeneous dose–response relationships. These features make it well suited to assess complex diet–health associations in high‐dimensional health data.

\section{Discussion}
This paper introduces a novel methodology for estimating heterogeneous causal effects in the presence of a high-dimensional and correlated treatment variable. Our contributions are threefold. First, we formalize a causal inference framework under exposure mapping, including the  mapping function properties, the causal estimands, and the assumption and demonstration of the estimand identification. Second, we propose a factor-model-based exposure mapping that reduces dimensionality of nutritional intake data and, consistent with common practice in nutritional epidemiology, discretizes each latent dietary pattern into three ordered exposure levels. Third, we develop an innovative extension of the Bayesian Causal Forest model, named BCF3L, that accommodates ordinal treatments by flexibly modeling the potential outcomes and estimating the causal effects for both medium versus low and high versus medium treatment levels.

 Motivated by a specific need in causal nutritional epidemiology,  the proposed approach and the BCF3L model are broadly applicable to other settings with three ordered treatment levels and complex, high-dimensional exposures.  
In addition, we develop an open-source \texttt{R} package, \texttt{BCF3L}, which implements our proposed nonparametric model and can be used in a wide range of applications.

Through a comprehensive simulation study, we demonstrate the ability of BCF3L to accurately recover heterogeneous causal effects across multiple data-generating scenarios,  including standard settings with an observed ordinal treatment and high-dimensional settings with latent factors requiring exposure mapping. Comparisons with Bayesian linear regression using horseshoe prior \citep{carvalho2010horseshoe} and with the BART extension for three treatment levels \citep{Yifan2025} show that BCF3L achieves consistently lower bias and mean square errors.

Our method applied to nutritional epidemiology data from the  HCHS/SOL study \citep{lavange2010sample}, allows us to understand the causal link between nutrition consumption and health risk factors. Using intake data on 53 nutrients from more than 9,000 Hispanic/Latino adults recruited across four U.S. field centers,  we identify six dietary patterns: plant lipid–antioxidant, dairy product, processed food, plant-based, animal proteins, and seafood. Overall, our findings highlight that increasing the intake from low to medium levels yields wider credible intervals and less statistical significance compared to increases from medium to high. Moreover, high consumption of plant lipid–antioxidant, plant-based, and dairy products is associated with significant reductions in BMI and fasting insulin, while moderate intake does not show clear effects. In contrast, moderate animal protein intake is more beneficial than high consumption. Processed food intake significantly increases BMI at high levels and seafood shows minimal causal effects, with the exception of an improvement in insulin levels under high consumption.

Several directions for future research follow naturally from this work. First, the exposure mapping component could be adapted to other applications that require structured dimensionality reduction, including longitudinal exposures \citep{jin2022bayesian} and settings with complex survey responses \citep{goplerud2025estimating}. Second, an important extension would incorporate uncertainty from the Bayesian factor model directly into the BCF3L stage, enabling fuller propagation of exposure-mapping uncertainty into treatment effect estimation.
%by incorporating the uncertainty of the Bayesian factor model directly into the BCF3L estimation step, improving the propagation of variability throughout the causal analysis. 
Finally, BCF3L could be extended to jointly model multiple ordered exposure variables, such as the simultaneous analysis of all dietary patterns, drawing inspiration from recent advances such as the VCBART model proposed by \citet{deshpande2020vcbart}.

\section*{Acknowledgements}
RDV was supported by the US National Institutes of Health, NIGMS/NIH COBRE CBHD P20GM109035  and `Programma per Giovani Ricercatori Rita Levi Montalcini'' granted by the Italian Ministry of Education, University, and Research. This manuscript was prepared using HCHS/SOL Research
Materials obtained from the NHLBI Biologic Specimen and Data Repository Information Coordinating Center and
does not necessarily reflect the opinions or views of the HCHS/SOL or the NHLBI.

%\newpage
%\appendix
%\renewcommand{\thesection}{S\arabic{section}}
%\setcounter{section}{0}
%\setcounter{table}{0}
%\setcounter{figure}{0}
%\renewcommand{\thetable}{S\arabic{table}}
%\renewcommand{\thefigure}{S\arabic{figure}}

\bibliographystyle{imsart-nameyear}
\bibliography{BCF}

@article{rubin1974estimating,
  title={Estimating causal effects of treatments in randomized and nonrandomized studies.},
  author={Rubin, Donald B},
  journal={Journal of Educational Psychology},
  volume={66},
  number={5},
  pages={688},
  year={1974},
  publisher={American Psychological Association}
}

@article{hahn2020bayesian,
  title={Bayesian regression tree models for causal inference: Regularization, confounding, and heterogeneous effects (with discussion)},
  author={Hahn, P Richard and Murray, Jared S and Carvalho, Carlos M},
  journal={Bayesian Analysis},
  volume={15},
  number={3},
  pages={965--1056},
  year={2020},
  publisher={International Society for Bayesian Analysis}
}

@article{Yifan2025,
  title={Diet and Cardiovascular Risk Factors: A Causal Inference Perspective Using Bayesian Additive Regression Trees with three treatment levell},
  author={Lu, Yifan and Zorzetto, Dafne and Edefonti, Valeria and De Vito, Roberta},
  journal={Work in progress},
  pages={ },
  year={2025},
  publisher={}
}

@article{carvalho2010horseshoe,
  title={The horseshoe estimator for sparse signals},
  author={Carvalho, Carlos M and Polson, Nicholas G and Scott, James G},
  journal={Biometrika},
  pages={465--480},
  year={2010},
  publisher={JSTOR}
}

@article{lichtenstein2003dietary,
  title={Dietary fat and cardiovascular disease risk: quantity or quality?},
  author={Lichtenstein, Alice H},
  journal={Journal of Women's Health},
  volume={12},
  number={2},
  pages={109--114},
  year={2003},
  publisher={Mary Ann Liebert, Inc.}
}

@article{islam2019trans,
  title={Trans fatty acids and lipid profile: A serious risk factor to cardiovascular disease, cancer and diabetes},
  author={Islam, Md Ashraful and Amin, Mohammad Nurul and Siddiqui, Shafayet Ahmed and Hossain, Md Parvez and Sultana, Farhana and Kabir, Md Ruhul},
  journal={Diabetes \& Metabolic Syndrome: Clinical Research \& Reviews},
  volume={13},
  number={2},
  pages={1643--1647},
  year={2019},
  publisher={Elsevier}
}

@article{castello2022adherence,
  title={Adherence to the Western, Prudent and Mediterranean Dietary Patterns and Colorectal Cancer Risk: Findings from the Spanish Cohort of the European Prospective Investigation into Cancer and Nutrition (EPIC-Spain)},
  author={Castell{\'o}, Adela  and others},
  journal={Nutrients},
  volume={14},
  number={15},
  pages={3085},
  year={2022},
  publisher={MDPI}
}

@article{edefonti2012nutrient,
  title={Nutrient-based dietary patterns and the risk of head and neck cancer: a pooled analysis in the International Head and Neck Cancer Epidemiology consortium},
  author={Edefonti, Valeria and others},
  journal={Annals of Oncology},
  volume={23},
  number={7},
  pages={1869--1880},
  year={2012},
  publisher={Elsevier}
}

@article{de2023multi,
  title={Multi-study factor regression model: an application in nutritional epidemiology},
  author={De Vito, Roberta and Avalos-Pacheco, Alejandra},
  journal={arXiv preprint arXiv:2304.13077},
  year={2023}
}

@article{kromhout2001diet,
  title={Diet and cardiovascular diseases.},
  author={Kromhout, Daan},
  journal={The journal of nutrition, health \& aging},
  volume={5},
  number={3},
  pages={144--149},
  year={2001}
}

@article{dominici2020controlled,
  title={From controlled to undisciplined data: estimating causal effects in the era of data science using a potential outcome framework},
  author={Dominici, Francesca and Bargagli-Stoffi, Falco J and Mealli, Fabrizia},
  journal={Harvard Data Science Review},
  year={2021}
}

@article{pearson2017reducing,
  title={Reducing US cardiovascular disease burden and disparities through national and targeted dietary policies: a modelling study},
  author={Pearson-Stuttard, Jonathan and Bandosz, Piotr and Rehm, Colin D and Penalvo, Jose and Whitsel, Laurie and Gaziano, Tom and Conrad, Zach and Wilde, Parke and Micha, Renata and Lloyd-Williams, Ffion and others},
  journal={PLoS medicine},
  volume={14},
  number={6},
  pages={e1002311},
  year={2017},
  publisher={Public Library of Science San Francisco, CA USA}
}

@article{gbd2015global,
  title={Global, regional, and national comparative risk assessment of 79 behavioural, environmental and occupational, and metabolic risks or clusters of risks in 188 countries, 1990--2013: a systematic analysis for the Global Burden of Disease Study 2013},
  author={{GBD 2013 Risk Factors Collaborators}},
  journal={Lancet},
  volume={386},
  number={10010},
  pages={2287},
  year={2015},
  publisher={Europe PMC Funders}
}

@article{deshpande2020vcbart,
  title={VCBART: Bayesian trees for varying coefficients},
  author={Deshpande, Sameer K and Bai, Ray and Balocchi, Cecilia and Starling, Jennifer E and Weiss, Jordan},
  journal={arXiv preprint arXiv:2003.06416},
  volume={2},
  pages={32--33},
  year={2020},
  publisher={Technical report}
}

@article{athey2019generalized,
  title={Generalized random forests},
  author={Athey, Susan and Tibshirani, Julie and Wager, Stefan},
  year={2019}
}

@article{bargagli2020causal,
  title={Causal rule ensemble: Interpretable discovery and inference of heterogeneous treatment effects},
  author={Bargagli-Stoffi, Falco J and Cadei, Riccardo and Lee, Kwonsang and Dominici, Francesca},
  journal={arXiv preprint arXiv:2009.09036},
  year={2020}
}

@article{alaa2017bayesian,
  title={Bayesian inference of individualized treatment effects using multi-task gaussian processes},
  author={Alaa, Ahmed M and Van Der Schaar, Mihaela},
  journal={Advances in neural information processing systems},
  volume={30},
  year={2017}
}

@article{zorzetto2024confounder,
  title={{Confounder-dependent Bayesian mixture model: Characterizing heterogeneity of causal effects in air pollution epidemiology}},
  author={Zorzetto, Dafne and Bargagli-Stoffi, Falco J and Canale, Antonio and Dominici, Francesca},
  journal={Biometrics},
  volume={80},
  number={2},
  year={2024}
}

@article{roy2018bayesian,
  title={Bayesian nonparametric generative models for causal inference with missing at random covariates},
  author={Roy, Jason and Lum, Kirsten J and Zeldow, Bret and Dworkin, Jordan D and Re III, Vincent Lo and Daniels, Michael J},
  journal={Biometrics},
  volume={74},
  number={4},
  pages={1193--1202},
  year={2018},
  publisher={Wiley Online Library}
}

@article{oganisian2021bayesian,
  title={A {B}ayesian nonparametric model for zero-inflated outcomes: Prediction, clustering, and causal estimation},
  author={Oganisian, Arman and Mitra, Nandita and Roy, Jason A},
  journal={Biometrics},
  volume={77},
  number={1},
  pages={125--135},
  year={2021},
  publisher={Wiley Online Library}
}

@article{hill2011bayesian,
  title={Bayesian nonparametric modeling for causal inference},
  author={Hill, Jennifer L},
  journal={Journal of Computational and Graphical Statistics},
  volume={20},
  number={1},
  pages={217--240},
  year={2011},
  publisher={Taylor \& Francis}
}

@article{linero2021and,
  title={The how and why of {B}ayesian nonparametric causal inference},
  author={Linero, Antonio R and Antonelli, Joseph L},
  journal={Wiley Interdisciplinary Reviews: Computational Statistics},
  volume={15},
  number={1},
  pages={e1583},
  year={2023},
  publisher={Wiley Online Library}
}

@article{maldonado2022posteriori,
  title={A Posteriori dietary patterns, insulin resistance, and diabetes risk by Hispanic/Latino heritage in the HCHS/SOL cohort},
  author={Maldonado, Luis E and Sotres-Alvarez, Daniela and Mattei, Josiemer and Daviglus, Martha L and Talavera, Gregory A and Perreira, Krista M and Van Horn, Linda and Mossavar-Rahmani, Yasmin and LeCroy, Madison N and Gallo, Linda C and others},
  journal={Nutrition \& Diabetes},
  volume={12},
  number={1},
  pages={44},
  year={2022},
  publisher={Nature Publishing Group UK London}
}

@article{lavange2010sample,
  title={{Sample design and cohort selection in the Hispanic Community Health Study/Study of Latinos}},
  author={LaVange, Lisa M  and others},
  journal={Annals of epidemiology},
  volume={20},
  number={8},
  pages={642--649},
  year={2010},
  publisher={Elsevier}
}

@article{de2022shared,
  title={Shared and ethnic background site-specific dietary patterns in the Hispanic Community Health Study/Study of Latinos (HCHS/SOL)},
  author={De Vito, Roberta and Stephenson, Briana and Sotres-Alvarez, Daniela and Siega-Riz, Anna-Maria and Mattei, Josiemer and Parpinel, Maria and Peters, Brandilyn A and Bainter, Sierra A and Daviglus, Martha L and Van Horn, Linda and others},
  journal={medRxiv},
  pages={2022--06},
  year={2022},
  publisher={Cold Spring Harbor Laboratory Press}
}

@article{thurstone1931multiple,
  title={Multiple factor analysis.},
  author={Thurstone, Louis Leon},
  journal={Psychological Review},
  volume={38},
  number={5},
  pages={406},
  year={1931},
  publisher={Psychological Review Company}
}

@article{anderson2004introduction,
  title={An introduction to multivariate statistical analysis},
  author={Anderson-Cook, Christine M},
  journal={Journal of the American Statistical Association},
  volume={99},
  number={467},
  pages={907--909},
  year={2004},
  publisher={American Statistical Association}
}

@article{Lopes2004,
  author = {Lopes, Hedibert Freitas and West, Mike},
  title = {Bayesian Model Assessment in Factor Analysis},
  journal = {Statistica Sinica},
  volume = {14},
  number = {1},
  year = {2004},
  pages = {41--67},
  publisher = {JSTOR}
}

@article{savje2024causal,
  title={Causal inference with misspecified exposure mappings: separating definitions and assumptions},
  author={S{\"a}vje, Fredrik},
  journal={Biometrika},
  volume={111},
  number={1},
  pages={1--15},
  year={2024},
  publisher={Oxford University Press}
}

@article{goplerud2025estimating,
  title={Estimating heterogeneous causal effects of high-dimensional treatments: Application to conjoint analysis},
  author={Goplerud, Max and Imai, Kosuke and Pashley, Nicole E},
  journal={The Annals of Applied Statistics},
  volume={19},
  number={2},
  pages={866--888},
  year={2025},
  publisher={Institute of Mathematical Statistics}
}

@article{jin2022bayesian,
  title={A Bayesian nonparametric approach for inferring drug combination effects on mental health in people with HIV},
  author={Jin, Wei and Ni, Yang and Rubin, Leah H and Spence, Amanda B and Xu, Yanxun},
  journal={Biometrics},
  volume={78},
  number={3},
  pages={988--1000},
  year={2022},
  publisher={Oxford University Press}
}

@article{hettinger2025multiply,
  title={Multiply robust difference-in-differences estimation of causal effect curves for continuous exposures},
  author={Hettinger, Gary and Lee, Youjin and Mitra, Nandita},
  journal={Biometrics},
  volume={81},
  number={1},
  pages={ujaf015},
  year={2025},
  publisher={Oxford University Press}
}

@article{mcjames2025bayesian,
  title={Bayesian causal forests for multivariate outcomes: application to Irish data from an international large scale education assessment},
  author={McJames, Nathan and O’Shea, Ann and Goh, Yong Chen and Parnell, Andrew},
  journal={Journal of the Royal Statistical Society Series A: Statistics in Society},
  volume={188},
  number={2},
  pages={428--450},
  year={2025},
  publisher={Oxford University Press UK}
}

@article{sun2025difference,
  title={Difference-in-Differences Under Network Interference},
  author={Sun, Kuan and Xiao, Zhiguo},
  journal={arXiv preprint arXiv:2509.24259},
  year={2025}
}

@article{negro2024sample,
  title={Sample distribution theory using coarea formula},
  author={Negro, Luigi},
  journal={Communications in Statistics-Theory and Methods},
  volume={53},
  number={5},
  pages={1864--1889},
  year={2024},
  publisher={Taylor \& Francis}
}

@article{Chipman_2010,
   title={BART: Bayesian additive regression trees},
   volume={4},
   ISSN={1932-6157},
   url={http://dx.doi.org/10.1214/09-AOAS285},
   DOI={10.1214/09-aoas285},
   number={1},
   journal={The Annals of Applied Statistics},
   publisher={Institute of Mathematical Statistics},
   author={Chipman, Hugh A. and George, Edward I. and McCulloch, Robert E.},
   year={2010},
   month=mar }

@article{de2019shared,
  title={Shared and study-specific dietary patterns and head and neck cancer risk in an international consortium},
  author={De Vito, R and Lee, Yuan Chin Amy and Parpinel, M and Serraino, D and Olshan, Andrew Fergus and Zevallos, Jose Pedro and Levi, F and Zhang, Zhuo Feng and Morgenstern, H and Garavello, W and others},
  journal={Epidemiology},
  volume={30},
  number={1},
  pages={93--102},
  year={2019},
  publisher={LWW}
}

@article{kaiser1960application,
  title={The application of electronic computers to factor analysis},
  author={Kaiser, Henry F},
  journal={Educational and Psychological Measurement},
  volume={20},
  number={1},
  pages={141--151},
  year={1960},
  publisher={Sage Publications Sage CA: Thousand Oaks, CA}
}

@article{samartsidis2019assessing,
  title={Assessing the causal effect of binary interventions from observational panel data with few treated units},
  author={Samartsidis, Pantelis and Seaman, Shaun R and Presanis, Anne M and Hickman, Matthew and De Angelis, Daniela},
  year={2019}
}

@article{samartsidis2020bayesian,
  title={A Bayesian multivariate factor analysis model for evaluating an intervention by using observational time series data on multiple outcomes},
  author={Samartsidis, Pantelis and Seaman, Shaun R and Montagna, Silvia and Charlett, Andr{\'e} and Hickman, Matthew and Angelis, Daniela De},
  journal={Journal of the Royal Statistical Society Series A: Statistics in Society},
  volume={183},
  number={4},
  pages={1437--1459},
  year={2020},
  publisher={Oxford University Press}
}

@article{zorzetto2025multivariate,
  title={Multivariate Causal Effects: a Bayesian Causal Regression Factor Model},
  author={Zorzetto, Dafne and Landy, Jenna and Zigler, Corwin and Parmigiani, Giovanni and De Vito, Roberta},
  journal={arXiv preprint arXiv:2504.03480},
  year={2025}
}

@article{varraso2012assessment,
  title={Assessment of dietary patterns in nutritional epidemiology: principal component analysis compared with confirmatory factor analysis},
  author={Varraso, Rapha{\"e}lle and Garcia-Aymerich, Judith and Monier, Florent and Le Moual, Nicole and De Batlle, Jordi and Miranda, Gemma and Pison, Christophe and Romieu, Isabelle and Kauffmann, Francine and Maccario, Jean},
  journal={The American journal of clinical nutrition},
  volume={96},
  number={5},
  pages={1079--1092},
  year={2012},
  publisher={Elsevier}
}

@article{huang2024sparse,
  title={Sparse Bayesian Factor Models with Mass-Nonlocal Factor Scores},
  author={Huang, Yingjie and Zorzetto, Dafne and De Vito, Roberta},
  journal={arXiv preprint arXiv:2412.00304},
  year={2024}
}

@article{bhattacharya2011,
  title={Sparse {B}ayesian infinite factor models},
  author={Bhattacharya, Anirban and Dunson, David B},
  journal={Biometrika},
  volume={98},
  number={2},
  pages={291},
  year={2011}
}

@article{rovckova2016,
  title={Fast {B}ayesian factor analysis via automatic rotations to sparsity},
  author={Ro{\v{c}}kov{\'a}, Veronika and George, Edward I},
  journal={Journal of the American Statistical Association},
  volume={111},
  number={516},
  pages={1608--1622},
  year={2016},
  publisher={Taylor \& Francis}
}

@article{Nanri2013,
  author  = {Nanri, A. and others},
  title   = {Dietary patterns and glucose tolerance abnormalities in Japanese men and women: the Japan Epidemiology Collaboration on Occupational Health Study},
  journal = {Diabetes Care},
  year    = {2013},
  volume  = {36},
  number  = {8},
  pages   = {2411--2419}
}

@article{Aune2018,
  author  = {Aune, D. and others},
  title   = {Whole grain consumption and risk of cardiovascular disease, cancer, and all cause and cause specific mortality: systematic review and dose-response meta-analysis of prospective studies},
  journal = {BMJ},
  year    = {2018},
  volume  = {361},
  pages   = {k2235}
}

@article{Zhao2020,
  author  = {Zhao, L. and others},
  title   = {Dietary patterns and chronic disease risk among Chinese adults: a prospective cohort study},
  journal = {Nutrients},
  year    = {2020},
  volume  = {12},
  number  = {10},
  pages   = {2985}
}

@article{Pittas2010,
  author  = {Pittas, A. G. and others},
  title   = {The role of vitamin D and calcium in type 2 diabetes. A systematic review and meta-analysis},
  journal = {Diabetes Care},
  year    = {2010},
  volume  = {33},
  number  = {6},
  pages   = {1379--1385}
}

@article{Song2013,
  author  = {Song, Y. and others},
  title   = {Dietary magnesium intake and risk of type 2 diabetes: meta-analysis of prospective cohort studies},
  journal = {American Journal of Clinical Nutrition},
  year    = {2013},
  volume  = {98},
  number  = {3},
  pages   = {774--788}
}

@article{Guo2020,
  author  = {Guo, W. and others},
  title   = {Dietary calcium intake and risk of type 2 diabetes: systematic review and dose–response meta-analysis},
  journal = {Nutrients},
  year    = {2020},
  volume  = {12},
  number  = {9},
  pages   = {2651}
}

@article{Mozaffarian2006,
  author  = {Mozaffarian, D. and others},
  title   = {Trans fatty acids and cardiovascular disease},
  journal = {New England Journal of Medicine},
  year    = {2006},
  volume  = {354},
  number  = {15},
  pages   = {1601--1613}
}

@article{Wang2021,
  author  = {Wang, D. D. and others},
  title   = {Association of long-term consumption of plant-based dietary patterns with cardiovascular disease risk},
  journal = {JAMA Cardiology},
  year    = {2021},
  volume  = {6},
  number  = {1},
  pages   = {98--108}
}

@article{McKeown2018,
  author  = {McKeown, N. M. and others},
  title   = {Whole-grain intake and insulin sensitivity: the Framingham Offspring Study},
  journal = {American Journal of Clinical Nutrition},
  year    = {2018},
  volume  = {107},
  number  = {1},
  pages   = {71--79}
}

@article{Kim2021,
  author  = {Kim, Y. and others},
  title   = {Whole grain consumption and risk of type 2 diabetes: a meta-analysis of prospective cohort studies},
  journal = {Nutrients},
  year    = {2021},
  volume  = {13},
  number  = {5},
  pages   = {1544}
}

@article{Weigle2005,
  author  = {Weigle, D. S. and others},
  title   = {A high-protein diet induces sustained reductions in appetite, ad libitum caloric intake, and body weight despite compensatory changes in diurnal plasma leptin and ghrelin concentrations},
  journal = {American Journal of Clinical Nutrition},
  year    = {2005},
  volume  = {82},
  number  = {1},
  pages   = {41--48}
}

@article{VanNielen2014,
  author  = {Van Nielen, M. and others},
  title   = {Dietary protein intake and incidence of type 2 diabetes in Europe: the EPIC-InterAct Case–Cohort Study},
  journal = {Diabetes Care},
  year    = {2014},
  volume  = {37},
  number  = {11},
  pages   = {3123--3130}
}

@article{Rietman2014,
  author  = {Rietman, A. and others},
  title   = {Protein intake, energy expenditure, and body composition in older adults},
  journal = {Journal of Nutrition},
  year    = {2014},
  volume  = {144},
  number  = {11},
  pages   = {1753--1763}
}

@article{Virtanen2014,
  author  = {Virtanen, J. K. and others},
  title   = {Serum omega-3 fatty acids and incident type 2 diabetes in men: the Kuopio Ischaemic Heart Disease Risk Factor Study},
  journal = {Diabetes Care},
  year    = {2014},
  volume  = {37},
  number  = {1},
  pages   = {189--196}
}

@article{Lankinen2018,
  author  = {Lankinen, M. and others},
  title   = {Dietary patterns and their associations with plasma lipid metabolites in the METSIM study},
  journal = {Nutrients},
  year    = {2018},
  volume  = {10},
  number  = {5},
  pages   = {611}
}

@article{Belalcazar2021,
  author  = {Belalcazar, L. M. and others},
  title   = {Lifestyle intervention for Hispanic/Latino adults with prediabetes: the HCHS/SOL Sociocultural Ancillary Study},
  journal = {Diabetes Care},
  year    = {2021},
  volume  = {44},
  number  = {2},
  pages   = {446--454}
}
\end{document}